\documentclass[final,5p,times,twocolumn,numbers,twoside]{elsarticle}

\usepackage[misc]{ifsym} 

\usepackage{amsthm}
\usepackage{amsmath,amssymb,amsfonts}
\usepackage{graphicx} 
\usepackage{float} 
\usepackage{epstopdf}
\usepackage{caption}
\usepackage{algorithmic}
\usepackage{textcomp}
\usepackage{xcolor}
\usepackage{subfigure}
\usepackage[justification=centering]{caption}
\usepackage{fancyhdr} 
\usepackage{geometry}

\usepackage{setspace} 

\usepackage{booktabs}
\usepackage{multirow}
\usepackage{tabularx}
\usepackage{longtable}
\usepackage{supertabular}
\def\BibTeX{{\rm B\kern-.05em{\sc i\kern-.025em b}\kern-.08em
    T\kern-.1667em\lower.7ex\hbox{E}\kern-.125emX}}

\newtheorem{definition}{Definition}
\newtheorem{theorem}{Theorem}
\newtheorem{corollary}{Corollary}

\newtheorem{remark}{Remark}

\newcommand{\sx}{\mkern-5mu}

\fancypagestyle{plain}{\fancyhf{}} 
\begin{document}

\begin{frontmatter}
\title{A Novel Sparse Sum and Difference Co-Array With Low Redundancy and Enhanced DOF for Non-Circular Signals}

\author[swu]{Si Wang}
\author[swu]{Guoqiang Xiao\corref{cor1}}
\address[swu]{College of Computer and Information Science, Southwest University, Chongqing 400715, China}
\cortext[cor1]{Corresponding author}
\ead{hw8888@email.swu.edu.cn\ (Si Wang), gqxiao@swu.edu.cn}

\begin{abstract}
Array structures based on the sum and difference co-arrays provide more degrees of freedom (DOF).
However, since the growth of DOF is limited by a single case of sum and difference co-arrays,
the paper aims to design a sparse linear array (SLA) with higher DOF via
exploring different cases of second-order cumulants.
We present a mathematical framework based on second-order cumulant to devise a second-order extended co-array (SO-ECA)
and define the redundancy of SO-ECA.
Based on SO-ECA, a novel array is proposed, namely low redundancy sum and difference array (LR-SDA),
which can provide closed-form expressions for the sensor positions
and enhance DOF in order to resolve more signal sources in the direction of arrival (DOA) estimation of non-circular (NC) signals.
For LR-SDA, the maximum DOF under the given number of total physical sensors can be derived and the SO-ECA of LR-SDA is hole-free.
Further, the corresponding necessary and sufficient conditions of signal reconstruction for LR-SDA are derived.
Additionally, the redundancy and weight function of LR-SDA are defined, and the lower band of the redundancy for LR-SDA is derived.
The proposed LR-SDA achieves higher DOF and lower redundancy than those of existing DCAs designed based on sum and difference co-arrays.
Numerical simulations are conducted to verify the superiority of LR-SDA on DOA estimation performance and enhanced DOF over other existing DCAs.
\end{abstract}


\end{frontmatter}


\section{Introduction}
Low cost sampling in intelligent perception is widely applied in many fields such as frequency estimation in the time domain \cite{Xiao2017notes}, \cite{Xiao2018robustness}, \cite{Xiao2016symmetric}, \cite{Xiao2021wrapped}, \cite{Xiao2023on}
and DOA estimation in the spatial domain.
This paper mainly focuses on DOA estimation of NC signals in array signal processing,
which is a fundamental problem studied for several decades \cite{Krim1996}, \cite{Godara1997}, \cite{Tuncer2009}, \cite{Xiao22023}.
It is well known that a uniform linear array (ULA) with $N$-sensors is used to estimate ($N-1$) sources using
DOA estimation methods such as MUSIC \cite{Schmidt1986} or ESPRIT \cite{Roy1989}.
To increase the DOF of ULA, more sensors are required, thus leading to a higher cost in practical applications.
However, nonuniform linear arrays (also known as SLAs ) offer an effective solution to these problems.
For an $N$-sensors sparse array,
the corresponding difference co-array (DCA) and sum co-array (SCA) can be constructed with the second-order cumulant (SOC) of the received NC signals,
which can provide $\mathcal{O}(N^2)$ consecutive lags, respectively \cite{Pal2010}, \cite{Pal2011}.
In this way, DOF can be increased significantly compared to traditional ULAs \cite{BD2017mutual}.

For the SLA research based on DCA,
minimum redundancy array (MRA) \cite{Moffet1968} is a foundational structure in order to obtain as large DOF as possible
by reducing redundant sensors.
However, the non-closed form expressions for sensor positions hinder MRA's scalability and complicate large-scale array design.
This limitation leads to extensive research on nested arrays \cite{Pal2010} and coprime arrays \cite{Pal2011},
which offer significant advantages with their closed-form expressions for the sensor locations.
The success of nested arrays and coprime arrays inspires further developments aimed at enhancing DOF,
including augmented coprime arrays \cite{Pal22011},
enhanced nested arrays \cite{Zhao2019} and arrays based on the maximum element spacing criterion \cite{Zheng2019}.
With the respect of DOA estimation,
traditional subspace-based methods \cite{Liu2015} only utilize the consecutive lags of the DCA, making a hole-free configuration advantageous.
Consequently, hole-filling strategies have been proposed to create new coprime arrays-like with a hole-free DCA \cite{Wang2019}.
Obviously, various DCAs based on SOC have been widely studied in DOA estimation because of its significantly enhanced DOF \cite{Pal2010}, \cite{Pal2011}, \cite{Zhao2019}, \cite{Zheng2019}, \cite{Shi2022}.

In addition, in recent years, numerous representative studies have focused on the exploration of SCA.
In \cite{Robin2017}, concatenated nested array (CNA) with a consecutive SCA is proposed,
which can provide more DOF with the same number of physical sensors compare to ULA.
Closed-form expressions of the physical sensors and number of unit inter-spacing between sensors in CNA are introduced.
Since the CNA uses significantly fewer sensors than the ULA provide the same DOF, considerable cost reductions may be achieved.
In \cite{Robin2020} , larger sparse array than CNA with consecutive SCA is proposed, namely Kl$\phi$ve array,
which can further decrease the redundancy than that of CNA.
A novel sparse linear array with consecutive SCA is proposed in \cite{Robin2021}, which is interleaved wichmann array.
There are only a few closely spaced sensors in interleaved wichmann array, which may make it more robust to mutual coupling effects.

Although SLAs designed based on DCA and SCA can greatly increase the number of  consecutive lags,
using a single SOC case to design SLAs still limits the increase in the number of consecutive lags for SLAs.
Therefore, in order to increase the number of consecutive lags for virtual array obtained by using SOC to a greater extent,
a framework of second-order extended co-array (SO-ECA) is proposed in the paper with combining four cases of SOC for received NC signals.
Furthermore, based on SO-ECA, we propose a low redundancy sum and difference array (LR-SDA) with hole-free co-arrays and closed-form expressions of sensor positions.
Additionally, to design a SLA structure with enhanced DOF by using different cases of SOC,
we introduce the criterion of forming more consecutive lags of designing SLA as follows.


\emph{Criterion 1 (Large consecutive lags of DCA and SCA)}:
The large consecutive lags of DCA and SCA are preferred, which can increase the number of resolvable sources in the DOA estimation \cite{Pal2010}, \cite{Piya2012}, \cite{Shen2019}.

\emph{Criterion 2 (Closed-form expressions of sensor positions)}:
A closed-form expression of sensor positions is preferred for scalability considerations \cite{Cohen2019}, \cite{Cohen2020}.

\emph{Criterion 3 (Hole-free DCA and SCA)}:
A SLA with a hole-free DCA and SCA is preferred,
since the data from its DCA and SCA can be utilized directly by subspace-based DOA estimation methods
which are easy to be implemented, and thus the algorithms based on compressive sensing \cite{Shen2015}, \cite{Shen20152}, \cite{Shen2017} or co-array interpolation techniques \cite{Cui2019}, \cite{Zhou2018} with increased computational complexity can be avoided \cite{Cohen2020}, \cite{Liu20172}.

\emph{Contribution:}
The paper focuses on the design of SLA in order to get hole-free SO-ECA based on SOC with enhanced DOF.
The main contributions of the paper are threefold.\par
$\bullet$ In the paper, an effective framework of SO-ECA is devised mathematically based on four cases of SOC for NC signals,
which can provide higher DOF than those of other SLAs designed based on DCA.
Further, the redundancy of SO-ECA is defined.
\par
$\bullet$ A novel LR-SDA is systematically designed based on SO-ECA, utilizing three ULAs placed side-by-side.
For the proposed LR-SDA with the given number of physical sensors,
the closed-form expressions of the physical sensor positions have been derived analytically,
and the DOF of LR-SDA is further enhanced by improving the configuration of the physical sensors among the three ULAs.
Consequently, the proposed LR-SDA offers significantly higher DOF and lower redundancy than those of other existing similar SLAs \cite{GuptaP2018}, \cite{WangY2020}, \cite{Yang2023}.
\par

$\bullet$ The necessary and sufficient conditions of signal reconstruction for LR-SDA are derived.
Additionally, the corresponding weight function and redundancy of LR-SDA are defined in the paper.

The paper is organized as follows. In Section II, we briefly introduce the general sparse array model.
In Section III, a mathematical framework is presented to derive a SO-ECA associated with four cases of SOC for NC signals.
The novel LR-SDA designed based on SO-ECA is proposed in Section IV,
which provides closed-form expressions for physical sensor locations by using SCA and DCA.
Furthermore, we explain that the SO-ECA of proposed LR-SDA is hole-free and calculate the corresponding maximum DOF.
And the necessary and sufficient conditions of signal reconstruction are derived in Section V.
The weight function and redundancy of LR-SDA are defined in Section VI and VII, respectively.
In Section VIII, several numerical simulations are presented to evaluate the RMSE of LR-SDA
compared to the other DCAs with respect to SNR, snapshots and the number of sources.

\textit{Notations:}
$\mathbb{S}$ is the physical sensor positions set of a SLA.
$N$ is the number of sensors. $D$ is the number of source signals to be estimated. $K$ is the number of snapshots.
$\Phi$ is sensor positions set of a SO-ECA.
The operators $\otimes$, $\odot$, $(\cdot)^T$, $(\cdot)^H$ and $(\cdot)^*$ stand for the Kronecker products,
Khatri-Rao products,
transpose, conjugate transpose and complex conjugation, respectively.
Set $\{a\sx : \sx b \sx:\sx c \}$ denotes the integer line from $a$ to $c$ sampled in steps of $b\in \mathbb{N}^+$.
When $b\sx=\sx1$, we use shorthand $\{a \sx:\sx c \}$ .

\section{Preliminaries}

\subsection{General Sparse Array Model}
Assume that there are $D$ non-Gaussian, non-circular and mutually uncorrelated far-field narrow band signals.
The incident angle of the $i^{th}$ signal is $\theta_i$,
and the physical sensor positions set of the SLA is represented as $\mathbb{S}=\{ p_1,p_2,\ldots,p_N \}\cdot d$,
where the unit spacing $d$ is generally set to half wavelength.
The array output at the $n^{th}$ physical sensor corresponding to the $t^{th}$ snapshot,
denoted as $x_n(t)$, can be expressed as follows
\begin{equation}
\label{w1}
\begin{aligned}
x_n(t)=\sum_{i=1}^Da_n(\theta_i)s_i(t)+n_n(t),
\end{aligned}
\end{equation}
where $a_n(\theta_i)$ denotes the steering response of $n^{th}$ physical sensor corresponding to the $i^{th}$ source signal,
which can be expressed as follows
\begin{equation}
\label{w8}
a_n(\theta_i)=e^{j\frac{2\pi p_{l_n} d}{\lambda}\sin(\theta_i)}.
\end{equation}

Further, $n_n(t)$ denotes a zero-mean additive Gaussian noise sample at the $n^{th}$ physical sensor, which is assumed to be
statistically independent of all the source signals.
Thus, the received signals from all $N$ physical sensors, denoted as $\boldsymbol{x}(t)=[x_1(t),\ldots,x_N(t)]^T$, can be expressed as follows
\begin{equation}
\label{w2}
\boldsymbol{x}(t)=\sum_{i=1}^D\boldsymbol{a}(\theta_i)s_i(t)+\boldsymbol{n}(t)=\boldsymbol{A}(\theta)\boldsymbol{s}(t)+\boldsymbol{n}(t),
\end{equation}
where $\boldsymbol{s}(t)=[s_1(t),\ldots,s_D(t)]^T$ denotes the source signals vector,
$\boldsymbol{n}(t)=[n_1(t),\ldots,n_N(t)]^T$ denotes the additive Gaussian noise vector,
and $\boldsymbol{a}(\theta_i)=[a_1(\theta_i),\ldots,a_N(\theta_D)]^T$ denotes the array steering vector corresponding
to the $i^{th}$ source signal and $\boldsymbol{A}(\theta)=[\boldsymbol{a}(\theta_1),\ldots,\boldsymbol{a}(\theta_D)]\in\mathbb{C}^{N\times D}$
denotes the array manifold matrix. Setting $p_1 = 0$, $\boldsymbol{a}(\theta_i)$ can be written as
\begin{equation}\nonumber
\boldsymbol{a}(\theta_i)=[1,e^{-j\frac{2\pi p_2d\sin\theta_i}{\lambda}},\ldots,e^{-j\frac{2\pi p_Nd\sin\theta_i}{\lambda}}]^T.
\end{equation}

For $K$ numbers of snapshots, (\ref{w2}) can be rewritten in matrix form as follows
\begin{equation}\nonumber
\boldsymbol{X}=\boldsymbol{A}\boldsymbol{S}+\boldsymbol{N},
\end{equation}
where $\boldsymbol{X}=[\boldsymbol{x}(1),..., \boldsymbol{x}(D)] \in\mathbb{C}^{N\times D}$ is the received signals matrix,
$\boldsymbol{S}=[\boldsymbol{s}(1),..., \boldsymbol{s}(D)] \in\mathbb{C}^{D\times K}$ is the source signals matrix, 
and $\boldsymbol{N}=[\boldsymbol{n}(1),..., \boldsymbol{n}(D)] \in\mathbb{C}^{N\times K}$ is the additive Gaussian noise matrix.

Before describing the proposed array structure, we firstly introduce the definitions  of SCA and DCA  for the completeness of this paper.

\begin{definition}
 (SCA \cite{Robin2017}): Given a physical sensor position set, $\mathbb{S }=\{p_{l_1}, p_{l_2},..., p_{l_N}\}\cdot d$, of a  SLA,
 the following set determines the sensor positions of a SCA
  \begin{equation}
 \Omega=\{ (p_{l_1}+p_{l_2}) d \ | \ l_1,l_2=1,...,N\},
 \label{wang2}
 \end{equation}
 where $p_{l_1}$ and $p_{l_2}$ represent the physical sensor positions given by the set $\mathbb{S}$.
 \end{definition}

\begin{definition}

(DCA \cite{Dias2017}): Given a physical sensor position set, $\mathbb{S} =\{p_{l_1}, p_{l_2},..., p_{l_N}\}\cdot d$, of a  SLA,
the following set determines the sensor positions of a DCA
\begin{equation}
\Gamma=\{ (p_{l_1}-p_{l_2}) d \ | \ l_1, l_2=1,2,...,N \},
\end{equation}
where $p_{l_1}$ and $p_{l_2}$ represent the physical sensor positions given by the set $\mathbb{S}$.
\end{definition}

In addition, to lighten the notations, given any two sets $\mathbb{S}$ and $\mathbb{S}'$ \cite{Xiao2023}, we use
\begin{equation}
\mathbb{C}(\mathbb{S},\mathbb{S}')=\{ p_i+p_j\ | \ p_i\in \mathbb{S},p_j\in \mathbb{S}' \},
\end{equation}
to denote the cross sum of elements from $\mathbb{S}$ and $\mathbb{S}'$.

\section{Second-Order Extended Co-Array}
The SOC of received NC signal vector $\boldsymbol{x}(t)$ in (\ref{w2}) can be represented as the following four cases

\begin{equation}
\begin{aligned}
\label{eq1}
&\mathcal{R}_{\boldsymbol{x}}^{(1)} \triangleq
E[\boldsymbol{x}(t)\boldsymbol{x}^T(t)]\\
&=\sum_{i=1}^D(\boldsymbol{a}(\theta_i) \boldsymbol{a}^T(\theta_i))E[s_i(t)s_i^T(t)]+ E[\boldsymbol{n}(t)\boldsymbol{n}^T(t)],\\
&\mathcal{R}_{\boldsymbol{x}}^{(2)}
\triangleq E[\boldsymbol{x}(t)\boldsymbol{x}^H(t)]\\
&= \sum_{i=1}^D(\boldsymbol{a}(\theta_i) \boldsymbol{a}^H(\theta_i))E[s_i(t)s_i^H(t)]+ E[\boldsymbol{n}(t)\boldsymbol{n}^H(t)],\\
&\mathcal{R}_{\boldsymbol{x}}^{(3)}
\triangleq E[\boldsymbol{x}^*(t)\boldsymbol{x}^T(t)]\\
&= \sum_{i=1}^D(\boldsymbol{a}^*(\theta_i) \boldsymbol{a}^T(\theta_i))E[s_i^*(t)s_i^T(t)]+ E[\boldsymbol{n}^*(t)\boldsymbol{n}^T(t)],\\
&\mathcal{R}_{\boldsymbol{x}}^{(4)} \triangleq
E[\boldsymbol{x}^*(t)\boldsymbol{x}^H(t)]\\
&= \sum_{i=1}^D(\boldsymbol{a}^*(\theta_i) \boldsymbol{a}^H(\theta_i))E[s_i^*(t)s_i^H(t)]+ E[\boldsymbol{n}^*(t)\boldsymbol{n}^H(t)].
\end{aligned}
\end{equation}


The vectorization of SOC $\mathcal{R}_{\boldsymbol{x}}^{(j)}$ generates the corresponding column vector $\boldsymbol{c}_{\boldsymbol{x}}^{(j)}\in\mathbb{C}^{N^2\times1}, j\in\{1,2,3,4\}$,
as follows

\begin{equation}
\label{wto18}
\begin{aligned}
\boldsymbol{c}_{\boldsymbol{x}}^{(j)}=vec(\mathcal{R}_{\boldsymbol{x}}^{(j)})=\sum_{i=1}^Db^{(j)}(\theta_i)p_{s_i}^{(j)}
=\boldsymbol{B}^{(j)}\boldsymbol{p}_{\boldsymbol{s}}^{(j)}+\boldsymbol{m}^{(j)},
\end{aligned}
\end{equation}
where $\boldsymbol{b}^{(j)}(\theta_i)\in \mathbb{C}^{N^2\times1}$, $\boldsymbol{B}^{(j)}\in\mathbb{C}^{N^2\times D}$,
$p_{s_i}^{(j)}\in\mathbb{C}$, $\boldsymbol{p}_{\boldsymbol{s}}^{(j)}\in \mathbb{C}^{D\times 1}$, $\boldsymbol{m}^{(j)}\in\mathbb{C}^{N^2\times 1}$.

Further, a new SO-ECA can be derived as follows by combining four $\boldsymbol{c}_{\boldsymbol{x}}^{(j)},j\in\{1,2,3,4\}$,
whose DOF are significantly enhanced compared to DCA

\begin{equation}
\label{w18}
\boldsymbol{c}_{\boldsymbol{x}}=[{\boldsymbol{c}_{\boldsymbol{x}}^{(1)}}^T,{\boldsymbol{c}_{\boldsymbol{x}}^{(2)}}^T,{\boldsymbol{c}_{\boldsymbol{x}}^{(3)}}^T,{\boldsymbol{c}_{\boldsymbol{x}}^{(4)}}^T]
\triangleq\boldsymbol{B}\boldsymbol{p}_s+\boldsymbol{m} \in \mathbb{C}^{4N^2\times1},
\end{equation}
where the equivalent source signal vector $\boldsymbol{p}_s$ is expressed as follows

\begin{equation}
\label{w22}
\boldsymbol{p}_s\triangleq[{\boldsymbol{p}_s^{(1)}}^T,{\boldsymbol{p}_s^{(2)}}^T,{\boldsymbol{p}_s^{(3)}}^T,{\boldsymbol{p}_s^{(4)}}^T]\in\mathbb{C}^{4D\times1},
\end{equation}
the equivalent noise vector $\boldsymbol{m}$ is
\begin{equation}
\boldsymbol{m}\triangleq[{\boldsymbol{m}^{(1)}}^T,{\boldsymbol{m}^{(2)}}^T,{\boldsymbol{m}^{(3)}}^T,{\boldsymbol{m}^{(4)}}^T] \in \mathbb{C}^{4N^2\times1},
\end{equation}
and the equivalent array manifold matrix $\boldsymbol{B}$ is
\begin{equation}
\label{w21}
\boldsymbol{B}= \left(
               \begin{matrix}
                 \boldsymbol{B}^{(1)} & 0                    & 0                    & 0\\
                 0                    & \boldsymbol{B}^{(2)} &0                     & 0\\
                 0                    & 0                    &\boldsymbol{B}^{(3)}  & 0\\
                 0                    & 0                    &0                     & \boldsymbol{B}^{(4)}\\
               \end{matrix}
             \right),
\end{equation}
where the specific expression of $\boldsymbol{B}^{(j)},\{j=1,2,3,4\}$ in (\ref{w21}) is as follows

\begin{equation}
\label{w6}
\begin{aligned}
&\boldsymbol{B}^{(j)}\triangleq[\boldsymbol{b}^{(j)}(\theta_1),\boldsymbol{b}^{(j)}(\theta_2),...,\boldsymbol{b}^{(j)}(\theta_D)],\\
&\boldsymbol{b}^{(j)}(\theta_i)\triangleq[{b}^{(j)}_1(\theta_i),{b}^{(j)}_2(\theta_i),...,{b}^{(j)}_{N^2}(\theta_i)]^T,(i\sx =\sx 1,2,...,D)\\
&\ \ \ \ \ \ =\begin{cases}
\boldsymbol{a}(\theta_i)\otimes\boldsymbol{a}(\theta_i), j=1,\\
\boldsymbol{a}(\theta_i)\otimes\boldsymbol{a}^*(\theta_i),j=2,\\
\boldsymbol{a}^*(\theta_i)\otimes\boldsymbol{a}(\theta_i),j=3,\\
\boldsymbol{a}^*(\theta_i)\otimes\boldsymbol{a}^*(\theta_i),j=4.
\end{cases}
\end{aligned}
\end{equation}

The SOCs for the four cases of the source are as follows
\begin{equation}
\label{w4}
\begin{aligned}
&\boldsymbol{p}_{\boldsymbol{s}}^{(j)}\triangleq[p_{s_1}^{(j)},p_{s_2}^{(j)},\ldots,p_{s_D}^{(j)}]^T,\\
&\begin{matrix}
p_{s_i}^{(j)}=\\
(i=1,2,\ldots,D)       \\
\end{matrix}
\begin{cases}
E[s_i(t)s_i(t)],j=1,\\
E[s_i(t)s_i^*(t)],j=2,\\
E[s_i^*(t)s_i(t)],j=3,\\
E[s_i^*(t)s_i^*(t)],j=4.
\end{cases}
\end{aligned}
\end{equation}

The SOCs for the four cases of the noise are as follows
\begin{equation}
\begin{aligned}
&\boldsymbol{m}\triangleq[\boldsymbol{0}, \sigma_n^2 \boldsymbol{I}_n, \sigma_n^2 \boldsymbol{I}_n,\boldsymbol{0}]^T\in\mathbb{C}^{4N^2\times1}.
\end{aligned}
\end{equation}

For $\boldsymbol{c}_{\boldsymbol{x}}^{(j)}=\boldsymbol{B}^{(j)}\boldsymbol{p}_{\boldsymbol{s}}^{(j)}+ {\boldsymbol{m}}^{(j)}, j\in\{1,2,3,4\}$ given in (\ref{wto18}),
it can be observed that $\boldsymbol{c}_{\boldsymbol{x}}^{(j)}$ is the result of vectorizing the SOC $\mathcal{R}_{\boldsymbol{x}}^{(j)}$,
where $\boldsymbol{p}_{\boldsymbol{s}}^{(j)}$ represents the equivalent source signals vector given in (\ref{w4}),
$\boldsymbol{B}^{(j)}$ represents the equivalent manifold matrix and $\boldsymbol{b}^{(j)}(\theta_i)$ represents
the equivalent steering vector given in (\ref{w6}) corresponding to the source signal.

Consequently, four virtual co-arrays under the four different cases can be obtained from
$\boldsymbol{c}_{\boldsymbol{x}}^{(j)}=\boldsymbol{B}^{(j)}\boldsymbol{p}_{\boldsymbol{s}}^{(j)}+ {\boldsymbol{m}}^{(j)}, j\in\{1,2,3,4\}$,
namely first second-order co-array (SOCA$_1$), second second-order co-array (SOCA$_2$), third second-order co-array (SOCA$_3$) and fourth second-order co-array (SOCA$_4$),
which are defined as follows.

\textbf{\textit{case 1:}} When $j=1$, we can get $\boldsymbol{b}^{(1)}(\theta_i)=\boldsymbol{a}(\theta_i)\otimes\boldsymbol{a}(\theta_i)$
from (\ref{w6}), and the elements of $\boldsymbol{b}^{(1)}(\theta_i)$ are given as follows
\begin{equation}
 \label{w11}
 \begin{aligned}
 &b^{(1)}_{N(l_1-1)+l_2}(\theta_i)
 =a_{l_1}(\theta_i)a_{l_2}(\theta_i)
 =e^{j\frac{2\pi d}{\lambda}(p_{l_1}+p_{l_2})\sin(\theta_i)}.
 \end{aligned}
 \end{equation}

Compared with the steering response $a_n(\theta_i)=e^{j\frac{2\pi p_{l_n} d}{\lambda}\sin(\theta_i)}$ for ULAs,
(\ref{w11}) implys the steering response of  sensor located at $(p_{l_1}+p_{l_2}) d$ for SOCA$_1$.
Consequently, the SOCA$_1$ derived from the vector $\boldsymbol{c}_{\boldsymbol{x}}^{(1)}$ in (\ref{w18})
can be considered as the virtual array, which is defined as follows.

\begin{definition}
(SOCA$_1$): For a linear array with N-sensors located at the given set $\mathbb{S}$, a multiset $\Phi_1$ is defined as follows
\begin{equation}
\label{w28}
\Phi_1\triangleq\{(p_{l_1}+p_{l_2}) d\ | \ l_1,l_2=1,2,\ldots,N\},
\end{equation}
where the multiset $\Phi_1$ allows repetitions, and has an underlying set $\Phi_1^{u}$ that contains the unique elements of $\Phi_1$.
Consequently, SOCA$_1$ is defined as the virtual linear array for case 1,
where the sensors are located at the positions given by the set $\Phi_1^u$.
\end{definition}

\textbf{\textit{case 2:}} When $j=2$, we can get $\boldsymbol{b}^{(2)}(\theta_i)=\boldsymbol{a}(\theta_i)\otimes\boldsymbol{a}^*(\theta_i)$
from (\ref{w6}), and the elements of $\boldsymbol{b}^{(2)}(\theta_i)$ are given as follows
\begin{equation}
 \label{wh8}
 \begin{aligned}
 &b^{(2)}_{N(l_1-1)+l_2}(\theta_i)
 =a_{l_1}(\theta_i)a_{l_2}^*(\theta_i)
 =e^{j\frac{2\pi d}{\lambda}(p_{l_1}-p_{l_2})\sin(\theta_i)}.
 \end{aligned}
 \end{equation}

Compared with the steering response $a_n(\theta_i)=e^{j\frac{2\pi p_{l_n} d}{\lambda}\sin(\theta_i)}$ for ULAs,
(\ref{wh8}) implys the steering response of  sensor located at $(p_{l_1}-p_{l_2}) d$ for SOCA$_2$.
Consequently, the SOCA$_2$ derived from the vector $\boldsymbol{c}_{\boldsymbol{x}}^{(2)}$ in (\ref{w18})
can be considered as the virtual array, which is defined as follows.

\begin{definition}
(SOCA$_2$): For a linear array with N-sensors located at the given set $\mathbb{S}$, a multiset $\Phi_2$ is defined as follows
\begin{equation}
\label{st12}
\Phi_2\triangleq\{(p_{l_1}-p_{l_2}) d\ | \ l_1,l_2=1,2,\ldots,N\},
\end{equation}
where the multiset $\Phi_2$ allows repetitions, and has an underlying set $\Phi_2^{u}$ that contains the unique elements of $\Phi_2$.
Consequently, SOCA$_2$ is defined as the virtual linear array for case 2,
where the sensors are located at the positions given by the set $\Phi_2^u$.
\end{definition}

\textbf{\textit{case 3:}} When $j=3$, we can get $\boldsymbol{b}^{(3)}(\theta_i)=\boldsymbol{a}^*(\theta_i)\otimes\boldsymbol{a}(\theta_i)$
from (\ref{w6}), and the elements of $\boldsymbol{b}^{(3)}(\theta_i)$ are given as follows
\begin{equation}
 \label{st11}
 \begin{aligned}
 &b^{(2)}_{N(l_1-1)+l_2}(\theta_i)
 =a_{l_1}^*(\theta_i)a_{l_2}(\theta_i)
 =e^{j\frac{2\pi d}{\lambda}(-p_{l_1}+p_{l_2})\sin(\theta_i)}.
 \end{aligned}
 \end{equation}

Compared with the steering response $a_n(\theta_i)=e^{j\frac{2\pi p_{l_n} d}{\lambda}\sin(\theta_i)}$ for ULAs,
(\ref{st11}) implys the steering response of  sensor located at $(-p_{l_1}+p_{l_2}) d$ for SOCA$_3$.
Consequently, the SOCA$_3$ derived from the vector $\boldsymbol{c}_{\boldsymbol{x}}^{(3)}$ in (\ref{w18})
can be considered as the virtual array, which is defined as follows.

\begin{definition}
(SOCA$_3$): For a linear array with N-sensors located at the given set $\mathbb{S}$, a multiset $\Phi_3$ is defined as follows
\begin{equation}
\label{w27}
\Phi_3\triangleq\{(-p_{l_1}+p_{l_2}) d\ | \ l_1,l_2=1,2,\ldots,N\},
\end{equation}
where the multiset $\Phi_3$ allows repetitions, and has an underlying set $\Phi_3^{u}$ that contains the unique elements of $\Phi_3$.
Consequently, SOCA$_3$ is defined as the virtual linear array for case 3,
where the sensors are located at the positions given by the set $\Phi_3^u$.
\end{definition}

\textbf{\textit{case 4:}} When $j=4$, we can get $\boldsymbol{b}^{(4)}(\theta_i)=\boldsymbol{a}^*(\theta_i)\otimes\boldsymbol{a}^*(\theta_i)$
from (\ref{w6}), and the elements of $\boldsymbol{b}^{(4)}(\theta_i)$ are given as follows
\begin{equation}
 \label{wh9}
 \begin{aligned}
 &b^{(4)}_{N(l_1-1)+l_2}(\theta_i)
 =a_{l_1}^*(\theta_i)a_{l_2}^*(\theta_i)
 =e^{j\frac{2\pi d}{\lambda}(-p_{l_1}-p_{l_2})\sin(\theta_i)}.
 \end{aligned}
 \end{equation}

Compared with the steering response $a_n(\theta_i)=e^{j\frac{2\pi p_{l_n} d}{\lambda}\sin(\theta_i)}$ for ULAs,
(\ref{wh9}) implys the steering response of  sensor located at $(-p_{l_1}-p_{l_2}) d$ for SOCA$_4$.
Consequently, the SOCA$_4$ derived from the vector $\boldsymbol{c}_{\boldsymbol{x}}^{(4)}$ in (\ref{w18})
can be considered as the virtual array, which is defined as follows.

\begin{definition}
(SOCA$_4$): For a linear array with N-sensors located at the given set $\mathbb{S}$, a multiset $\Phi_3$ is defined as follows
\begin{equation}
\label{w30}
\Phi_4\triangleq\{(-p_{l_1}-p_{l_2}) d\ | \ l_1,l_2=1,2,\ldots,N\},
\end{equation}
where the multiset $\Phi_4$ allows repetitions, and has an underlying set $\Phi_4^{u}$ that contains the unique elements of $\Phi_3$.
Consequently, SOCA$_4$ is defined as the virtual linear array for case 4,
where the sensors are located at the positions given by the set $\Phi_4^u$.
\end{definition}


Furthermore, $\boldsymbol{c}_{\boldsymbol{x}}\in\mathbb{C}^{4N^2\times1}$ in (\ref{w18}) is obtained by
combing four $\boldsymbol{c}_{\boldsymbol{x}}^{(j)}, j\in\{1,2,3,4\}$.
It is equivalent to the cumulants of the signals received in a single snapshot by the constructed virtual linear array,
which is the combination of all four possible co-arrays for $j\in \{1, 2, 3,4\}$, namely SOCA$_1$, SOCA$_2$, SOCA$_3$ and SOCA$_4$.
Therefore, the derived virtual linear array is the SO-ECA,
and to obtain the sensor position set of SO-ECA, we first introduce the following knowledge about multiset.

A multiset $\Phi$ is defined as the multiset-sum (bag sum) of the four multisets $\Phi= \Phi_1 + \Phi_2 + \Phi_3+ \Phi_4$,
which denotes the union operation of the set subjected to sets with duplicate elements, and the specific information is in \cite{Wang2024}.
On the contrary, for the multiset $\Phi$, there exists a set $\Phi^u=\Phi_1^u \cup \Phi_2^u \cup \Phi_3^u\cup \Phi_4^u$ that contains the unique elements.

After obtaining the set of SO-ECA physical sensors, the SO-ECA is defined as follows.
\begin{definition}
(SO-ECA):
For a linear array of N-sensors located at positions given by the set $\mathbb{S}$,
the SO-ECA is derived based on the SOCs, whose sensors are located at the set $\Phi^u$.
\end{definition}

Therefore, the SO-ECA is defined as the virtual linear array contained $|\Phi^u|=\mathcal{O}(4N^2)$ sensors.
In addition, it worths to notes  the following two remarks in order to fully understand SO-ECA.

\begin{remark}
$\Phi_2$ in (\ref{st12}) and $\Phi_3$ in (\ref{w27}) is symmetric corresponding to zero.
$\Phi_1$ in (\ref{w28}) and $\Phi_4$ in (\ref{w30}) are opposite numbers to each other.
Consequently, the sensor position set $\Phi^u$ of SO-ECA is the union of the four sets $\Phi_1^u$, $\Phi_2^u$, $\Phi_3^u$ and $\Phi_4^u$,
which is symmetric corresponding to zero.
It means that if a virtual sensor locate at $p \in \Phi^u\cdot d$, there must exist another corresponding virtual sensor located at $-p \in \Phi^u\cdot d$.
\end{remark}

\begin{remark}
In general, the SO-ECA of an arbitrary linear array might not be a hole-free array.
For example, the SO-ECA of a linear array with sensor positions given by $\mathbb{S} = \{0,1,5,8\}\cdot d$ can be obtained based on SOCs,
and the virtual sensor positions on the non-negative side are given by
$\{0,1,2,3,4,5,6,7,8,9,10,\text{X},\text{X},13,\text{X},\text{X},16 \}\cdot d$ without 11d, 12d, 14d, 15d.
\end{remark}

\section{Low Redundancy Sum and Difference Array Based on SO-ECA}
When a SO-ECA is hole-free, it can be easily utilized to estimate DOA without any spatial aliasing \cite{Pal2010}, \cite{Pal22011}, \cite{Piya2012}.
Moreover, the number of consecutive lags of DCA mainly
depends on the physical sensor geometry of a linear array \cite{Cohen2020}.
Therefore, the low redundancy sum and difference array (LR-SDA) is proposed based on SO-ECA in the paper to enhance the DOF and reduce the redundancy.
The proposed LR-SDA can be designed systematically by appropriately placing the three sub-arrays as shown in Fig. 2,
where sub-array 1 is a ULA with a big inter-spacing between sensors,
and sub-array 2 and 3 are ULA with unit inter-spacing between sensors.

\subsection{Structure of the Proposed LR-SDA}
\begin{figure}[h]
 \center{\includegraphics[width=8cm]  {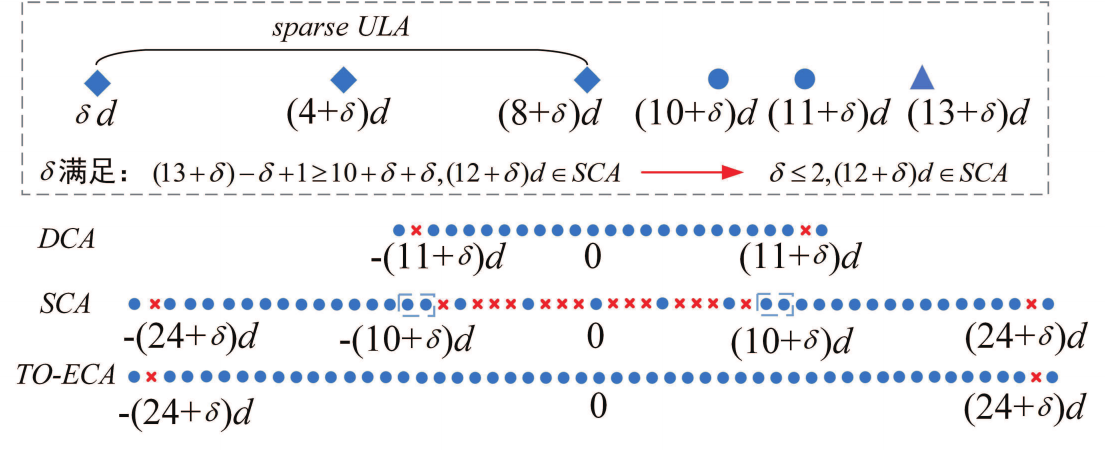}}
 \caption{Co-arrays with 6 sensors LR-SDA.}
\end{figure}

According to the CNA, to obtain consecutive sum co-array,
the right part of the array consists of sensors with an unit inter-spacing.
Furthermore, in SO-ECA, virtual sensor positions can be obtained through difference or sum of any two physical sensors.
Therefore, the first sensor position of the array does not necessarily start from zero,
allowing the array to be shifted to the right by $\delta$.
Additionally, to reduce the overlapping part between DCA and SCA in SO-ECA,
the right part of the array with an unit inter-spacing can be split into two sections, thereby increasing the inter-spacing between sensors.
The specific example for $N=6$ physical sensors is shown in Fig. 1,
where the DCA with consecutive range $\{ -11-\delta:11+\delta\} \cdot d$,
and the SCA with positive consecutive range $\{ 10+\delta:24+\delta\} \cdot d$,
further resulting sum-difference co-array with consecutive range $\{ -24-\delta:24+\delta\} \cdot d$,
which  possesses more virtual sensors than $\{-22:22\}\cdot d$ of TNA \cite{WangY2020} and can also be directly utilized to perform DOA estimation.
Next, we propose the detail definition of LR-SDA as follows.

\begin{definition}
(LR-SDA) The LR-SDA consists of three sub-arrays with the number of physical sensors $N=N_1+(N_2-\eta)+\eta$,
where $N_1$, $N_2-\eta$ and $\eta$ represent the number of physical sensors in each sub-array.
These sensors in LR-SDA are located at positions given by the set $\mathbb{S}_1$, $\mathbb{S}_2$ and $\mathbb{S}_3$,
which can be represented as follows, and the structure of the LR-SDA is shown as Fig. 2.
\begin{equation}\nonumber
\label{st1}
\begin{aligned}
&\mathbb{S}=\mathbb{S}_1\cup\mathbb{S}_2\cup\mathbb{S}_3,\\
&\mathbb{S}_1=\{ \delta:N_2+1:(N_1-1)(N_2+1)+\delta \}\cdot d,\\
&\mathbb{S}_2=\{ (N_1\sx -\sx 1)(N_2\sx +\sx 1)\sx +\sx \eta \sx +\sx 1\sx +\sx \delta:(N_1\sx -\sx 1)(N_2\sx +\sx 1)\sx +\sx N_2\sx +\sx \delta \}\sx \cdot\sx d,\\
&\mathbb{S}_3=\{ N_1(N_2+1)+1+\delta:N_1(N_2+1)+\eta+\delta \}\cdot d,\\
&\eta = \lceil \frac{N_2}{2} \rceil - 1,\\
&\delta = \lfloor \frac{N_2 + 1}{2} \sx \rfloor \ and \  (N_1\sx -\sx 1)(N_2\sx +\sx 1)\sx +\sx N_2+ \delta\sx +\sx 1 \sx \in \sx \mathbb{C}(\mathbb{S}_1,\mathbb{S}_1).
\end{aligned}
\end{equation}
\end{definition}

\begin{figure*}
 \center{\includegraphics[width=15cm]  {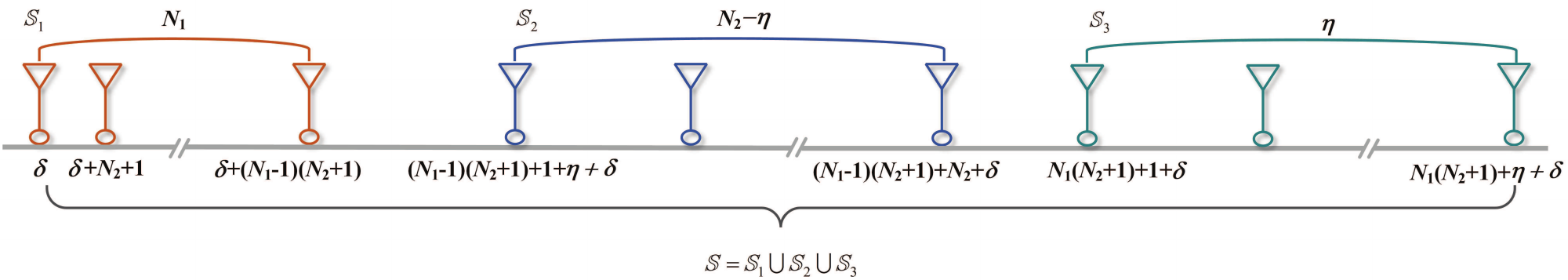}}
 \caption{\label{1} Structure of LR-SDA}
\end{figure*}

\begin{theorem}
The SO-ECA of LR-SDA is hole-free.
\end{theorem}
\begin{proof}
Since the consecutive lags of SO-ECA for LR-SDA are symmetry corresponding to zero,
only the positive lags are considered as follows.
Firstly, we can get the cross sum and self sum of $\mathbb{S}_1$, $\mathbb{S}_2$ and $\mathbb{S}_3$  as follows
\begin{equation}\nonumber
\begin{aligned}
&\mathbb{C}(\mathbb{S}_1,\mathbb{S}_1)=\{ 2\delta:(N_1+1):2(N_1-1)(N_2+1)+2\delta \}, \\
&\mathbb{C}(\mathbb{S}_1,\mathbb{S}_2)\sx =\sx \{ 2\delta\sx +\sx (N_1\sx -\sx 1)(N_2\sx +\sx 1)\sx +\eta\sx +\sx 1:N_1(N_2\sx +\sx 1)\sx +\sx 2\delta\sx -\sx 1\}\\
&\cup\{ N_1(N_2+1)+2\delta+\eta+1:N_1(N_2+1)+2\delta\}\cup ... \\
&\cup\{2N_1(N_2\sx +\sx 1)\sx -\sx 2N_2\sx -\sx 1\sx +\sx \eta\sx +\sx 2\delta :2N_1(N_2\sx +\sx 1)-N_2\sx -\sx 2\sx +\sx 2\delta\}, \\
&\mathbb{C}(\mathbb{S}_1,\mathbb{S}_3)=\{N_1(N_2\sx +\sx 1)\sx +\sx 1\sx +\sx 2\delta:N_1(N_2\sx +\sx 1)\sx +\sx \eta\sx +\sx 2\delta \}\cup \\
&\{ N_1(N_2\sx +\sx 1)\sx +\sx 2\sx +\sx N_1\sx +\sx 2\delta:N_1(N_2\sx +\sx 1)\sx +\sx \eta\sx +\sx 1\sx +\sx 2\delta \}\cup ... \cup \\
&\{ 2N_1(N_2\sx +\sx 1)\sx +\sx 2\delta\sx +\sx 2\eta\sx -\sx N_2\sx :\sx 2N_1(N_2\sx +\sx 1)\sx +\sx \eta\sx +\sx 2\delta\sx -\sx N_2\sx -\sx 1 \},\\
&\mathbb{C}(\mathbb{S}_2,\mathbb{S}_2)\sx =\sx \{ 2N_1(N_2\sx +\sx 1)\sx +\sx 2\eta\sx +\sx 2\delta\sx -\sx 2N_2\sx :\sx 2N_1(N_2\sx +\sx 1)\sx +\sx 2\delta\sx -\sx 2 \},\\
&\mathbb{C}(\mathbb{S}_2,\mathbb{S}_3)\sx =\sx \{ 2N_1(N_2\sx +\sx 1)\sx +\sx \eta\sx +\sx 1\sx +\sx 2\delta\sx -\sx N_2\sx :\sx 2N_1(N_2\sx +\sx 1)\sx +\sx \eta\sx -\sx 1\sx +\sx 2\delta \},\\
&\mathbb{C}(\mathbb{S}_3,\mathbb{S}_3)\sx =\sx \{2N_1(N_2\sx +\sx 1)+2+2\delta\sx :\sx 2N_1(N_2+1)+2\eta+2\delta\}.
\end{aligned}
\end{equation}

Secondly, we can get the cross difference and self difference of $\mathbb{S}_1$, $\mathbb{S}_2$ and $\mathbb{S}_3$  as follows
\begin{equation}\nonumber
\begin{aligned}
&\mathbb{C}(\mathbb{S}_1,-\mathbb{S}_1)=\{ 0:N_1+1:(N_1-1)(N_2+1)\},\\
&\mathbb{C}(-\mathbb{S}_1,\mathbb{S}_2)=\{ \eta\sx +\sx 1:N_2\} \cup \{ \eta\sx +\sx N_1\sx +\sx 2 :N_2\sx +\sx N_1\sx +\sx 1 \}\cup...\\
&\cup\{(N_1-1)(N_2+1)+\eta+1: N_1(N_2+1)-1 \},\\
&\mathbb{C}(-\mathbb{S}_1,\mathbb{S}_3)=\{N_2\sx +\sx 2\sx :\sx N_2\sx +\sx 1\sx +\sx \eta \}\sx \cup \sx \{N_2\sx +\sx N_1\sx +\sx 3\sx :\sx N_2\sx +\sx N_1\sx +\sx \eta\sx +\sx 2\}\\
&\cup...\cup\{ N_1(N_2+1)+1:N_1(N_2+1)\sx +\sx \eta \},\\
&\mathbb{C}(\mathbb{S}_2,-\mathbb{S}_2)=\{ \eta+1-N_2:N_2-\eta-1 \},\\
&\mathbb{C}(-\mathbb{S}_2,\mathbb{S}_3)=\{ 2:N_2\},\\
&\mathbb{C}(\mathbb{S}_3,-\mathbb{S}_3)=\{1-\eta:\eta-1\}.
\end{aligned}
\end{equation}

Thus, the union of all sets above is
\begin{equation}
\begin{aligned}
&\mathbb{C}(\mathbb{S}_1,\mathbb{S}_1)\cup\mathbb{C}(\mathbb{S}_1,\mathbb{S}_2)\cup\mathbb{C}(\mathbb{S}_1,\mathbb{S}_3)
\cup\mathbb{C}(\mathbb{S}_2,\mathbb{S}_2)\cup\mathbb{C}(\mathbb{S}_2,\mathbb{S}_3)\\
&\cup\mathbb{C}(\mathbb{S}_2,\mathbb{S}_3)\cup\mathbb{C}(\mathbb{S}_3,\mathbb{S}_3)\cup\mathbb{C}(\mathbb{S}_1,-\mathbb{S}_1)
\cup\mathbb{C}(-\mathbb{S}_1,\mathbb{S}_2)\\
&\cup\mathbb{C}(-\mathbb{S}_1,\mathbb{S}_3)\cup\mathbb{C}(\mathbb{S}_2,-\mathbb{S}_2)\cup\mathbb{C}(-\mathbb{S}_2,\mathbb{S}_3)
\cup\mathbb{C}(\mathbb{S}_3,-\mathbb{S}_3)\\
&=\{ 0:2N_1(N_2+1)+2\eta+2\delta \}.
\end{aligned}
\end{equation}

To sum up, the SO-ECA of proposed LR-SDA is hole-free,
and the DOF of LR-SDA is $4N_1(N_2+1)+4\eta+4\delta+1$.

\end{proof}

\subsection{The Maximum DOF of LR-SDA With the Given Number of Physical Sensors}
The DOF of LR-SDA designed based on SO-ECA can be further increased by optimizing the distribution of the physical sensors
among three sub-arrays under a given number of physical sensors.

\begin{theorem}
To obtain the maximum DOF of LR-SDA with the given number of physical sensors,
the number of sensors in the three sub-arrays is set to
\begin{equation}
\begin{cases}
N_1=N-N_2,\\
N_2=\lceil\frac{N-1}{2}\rfloor-\eta,\\
N_3=\eta,
\end{cases} \sx (\eta=1)\ or \
\begin{cases}
N_1=N-N_2,\\
N_2=\lceil\frac{2N-1}{4}\rfloor-\eta,\\
N_3=\eta,
\end{cases} \sx (\eta\geq2).
\end{equation}
\end{theorem}

\begin{proof}
Firstly, it follows from Definition 8 that the DOF of the LR-SDA with $\eta=1$ are
\begin{equation}
\label{st2}
\begin{aligned}
\text{DOF}&=4N_1N_2+4N_1+4\delta+1.
\end{aligned}
\end{equation}

Furthermore, the total number of sensors is
\begin{equation}
\label{st3}
N=N_1+N_2.
\end{equation}

Eq. (\ref{st2}) and (\ref{st3}) can be used to formulate an optimization problem for finding parameters $N_1$ and
$N_2$ that maximize the DOF of LR-SDA with a given number of sensors $N$
\begin{equation}\nonumber
\begin{aligned}
&\underset{N_1,N_2\in \mathbb{N}_+}{maximize}\ \ 4N_1N_2+4N_1+4\delta+1,\\
&subject\ to\ N_1+N_2=N. \ \ \ \ \ \ \ \ \ \ \ \ \ \ \ \ \ \ \ \ \ \ \ \ \ \ \ \ \ (P1)
\end{aligned}
\end{equation}

\subsubsection{Solution to relaxed problem}
~\par
Although (P1) is an problem with no general closed-form integer solution,
the maximum possible DOF maybe found under the relaxation that $N_1, N_2\in \mathbb{R}_+$.
Solving (\ref{st3}) for $N_2$ and substituting the result into (\ref{st2}) yields

\begin{equation}\nonumber
\text{DOF}=-4N_2^2+(4N-4)N_2+4\eta+4\delta+1.
\end{equation}

Since DOF is a concave function about $N_2$, we can straightly solve the maximum DOF \cite{Hiriart1989}.
The optimal parameter pair of the relaxed problem is
\begin{equation}\nonumber
N_2^*=\frac{N-1}{2}.
\end{equation}

In addition, since $N_2\in \mathbb{Z}$, we obtain $N_2$ as
\begin{equation}\nonumber
N_2^*=\lceil \frac{N-1}{2} \rfloor.
\end{equation}

Secondly, it follows from Definition 8 that the DOF of the LR-SDA with $\eta\geq2$ are
\begin{equation}
\label{st4}
\begin{aligned}
\text{DOF}&=4N_1N_2+4N_1+4\eta+4\delta+1.
\end{aligned}
\end{equation}

Furthermore, the total number of sensors is
\begin{equation}
\label{st5}
N=N_1+N_2.
\end{equation}

Eq. (\ref{st5}) and (\ref{st4}) can be used to formulate an optimization problem for finding parameters $N_1$,
$N_2$ that maximize the DOF of LR-SDA with a given number of sensors $N$
\begin{equation}\nonumber
\begin{aligned}
&\underset{N_1,N_2\in \mathbb{N}_+}{maximize}\ \ 4N_1N_2+4N_1+4\eta+4\delta+1,\\
&subject\ to\ N_1+N_2=N. \ \ \ \ \ \ \ \ \ \ \ \ \ \ \ \ \ \ \ \ \ \ \ \ \ \ \ \ \ (P2)
\end{aligned}
\end{equation}

Although (P2) is an problem with no general closed-form integer solution,
the maximum possible DOF maybe found under the relaxation that $N_1, N_2\in \mathbb{R}_+$.
Solving (\ref{st5}) for $N_2$ and substituting the result into (\ref{st4}) yields

\begin{equation}\nonumber
\text{DOF}=-4N_2^2+(4N-2)N_2+4N+4\delta+1.
\end{equation}

Since DOF is a concave function about $N_2$, we can straightly solve the maximum DOF \cite{Hiriart1989}.
The optimal parameter pair of the relaxed problem is
\begin{equation}\nonumber
N_2^*=\frac{2N-1}{4}.
\end{equation}

In addition, since $N_2\in \mathbb{Z}$, we obtain $N_2$ as
\begin{equation}\nonumber
N_2^*=\lceil \frac{2N-1}{4} \rfloor.
\end{equation}
\end{proof}

\subsubsection{General Solution}
~\par

When a solution for a specific maximum positive consecutive lags $E$ is required,
the parameters $N_1$ and $N_2$ of LR-SDA can be found by minimizing the number of sensors for a target $E$ instead
\begin{equation}\nonumber
\begin{aligned}
&\underset{N_1,N_2\in \mathbb{N}_+}{minimize}\ \ N=N_1+N_2,\\
&subject\ to\
\begin{cases}
2N_1N_2+2N_1+2\delta=E, (\eta=1),\\
2N_1N_2+2N_1+2\eta+2\delta=E, (\eta \geq 2).\ (P3)
\end{cases}
\end{aligned}
\end{equation}

Relaxing the integer constraint to $N_1,N_2 \in\mathbb{R}_+$ and solving (P3) for $N_1$ yields
\begin{equation}
\label{wh7}
\begin{aligned}
&N_1^*=\frac{\sqrt{2}\sqrt{E-2\delta}}{2},\ (\eta=1),\\
&N_1^*=\frac{-1+\sqrt{2}\sqrt{E+3-2\delta}}{2},\ (\eta\geq 2).
\end{aligned}
\end{equation}

\textbf{Note:} The specific derivation process of $N_1^*$ in (\ref{wh7}) is shown in Appendix A.

\section{Necessary and Sufficient Conditions for Signal Reconstruction}

In array signal processing, the problem of signal reconstruction can be fundamentally reformulated as the estimation of the DOA.
Specifically, for the received signal $\mathbf{x}(t)$ with an $N$-sensors array denoted as (\ref{w2}),
and $\mathbf{s}(t) \in \mathbb{C}^D$ is the source signals vector from $D$ narrowband far-field sources,
$\boldsymbol{\theta}_{[1:D]} = [\theta_1, \dots, \theta_D]^T $ denoting the unknown DOAs,
and $\mathbf{A}(\boldsymbol{\theta}) = [\mathbf{a}(\theta_1), \dots, \mathbf{a}(\theta_D)] \in \mathbb{C}^{N \times D} $ is the array manifold matrix,
composed of steering vectors determined by the array geometry and $\boldsymbol{\theta}_{[1:D]}$.
Once $\boldsymbol{\theta}_{[1:D]}$ are accurately estimated
and the array manifold matrix $\mathbf{A}(\boldsymbol{\theta})$ satisfies the rank condition,
the source signals vector $\mathbf{s}(t)$ can be reconstructed from the received signal $\mathbf{x}(t)$.
At this point, whether the signal can be reconstructed is equivalent to whether there is a one-to-one correspondence between $\mathbf{x}(t) $ and $\boldsymbol{\theta}_{[1:D]}$.
This requires several conditions are satisfied as follows \cite{Tuncer2009}, \cite{ThompsonAR2017}, \cite{Krim1996}.

Firstly, the signal characteristics must be taken into account.
The signal should satisfy the narrowband assumption, i.e., $B \ll f_c $,
where $B$ is the signal bandwidth and $f_c$ is the carrier frequency.
Secondly, the relationship between the array dimension and the number of source signals is crucial.
The number of sensors $N$ should be greater than or equal to the number of source signals $D$, i.e., $N \geq D$.
Moreover, the array manifold matrix $\mathbf{A}$ should have full column rank, i.e., $\text{rank}(\mathbf{A}) = D $ \cite{Horn2012}.
Under above conditions, the linear system (\ref{w2}) ensures a unique solution, establishing the one-to-one correspondence between $\mathbf{x}(t)$  and $\mathbf{\theta}_{[1:D]}$.
Finally, the array geometry plays a vital role.
When the inter-spacing $d$ exceeds half the wavelength ($\lambda/2$) in any array,
the problem of ambiguous angles may arise,
causing the array steering matrix $\mathbf{A}$ to lose full rank.
This leads to a situation where the one-to-one correspondence between $\mathbf{x}(t)$ and $\boldsymbol{\theta}_{[1:D]}$ is no longer satisfied,
making it impossible to uniquely reconstruct the source signals.
Therefore, the geometry and the inter-spacing $d$ should satisfy a certain threshold to prevent the occurrence of grid flap problems,
ensuring that DOA estimation remains unambiguous.

To sum up, there exists necessary and sufficient conditions for signal reconstruction as follows.

\begin{theorem}
For any array $\{ p_1, p_2, \dots, p_N \} \cdot d$, the sufficient and necessary condition for signal reconstruction is
\begin{equation}
\label{st14}
\operatorname{lcm}\left( \frac{\lambda}{p_{l_1}}, \frac{\lambda}{p_{l_2}}, \dots, \frac{\lambda}{p_{l_N}} \right) \geq 2,
\end{equation}
that is,
\[
\min_{c_1, c_2, \dots, c_N \in \mathbb{Z}^+} \left\{ \frac{c_1 \lambda}{p_{l_1}} = \frac{c_2 \lambda}{p_{l_2}} = \dots = \frac{c_N \lambda}{p_{l_N}} \right\} \geq 2.
\]
\end{theorem}

\begin{proof}
\textbf{Sufficiency:}
When (\ref{st14}) is satisfied,
it ensures that the array can uniquely resolve all directions within the angular field of view, without aliasing.
Therefore, for any $\theta_i \in [-\frac{\pi}{2}, \frac{\pi}{2}]$, the mapping from $\theta_i$ to the steering vector $\mathbf{a}(\theta_i)$ is injective.

Hence, no two distinct $\theta_i \neq \theta_i'$ can yield the same steering vector, i.e., $\mathbf{a}(\theta_i) \neq \mathbf{a}(\theta_i')$,
and the array manifold $\mathbf{A}$ is uniquely determined by $\mathbf{\theta}_{[1:D]}$.
Therefore, there is a a unique correspondence between the received signal vector $\mathbf{x}(t)$ and $\mathbf{\theta}_{[1:D]}$
when a full-rank signal matrix $\mathbf{S}$ is given,
ensuring the source signals $\mathbf{s}(t)$ can be reconstructed from the received signal $\mathbf{x}(t)$.

\textbf{Necessity:}
When the signal is uniquely reconstructable, i.e., there exists a bijective mapping between $\mathbf{x}(t)$ and $\mathbf{\theta}_{[1:D]}$,
which implies that the steering vector satisfies
\[
\mathbf{a}(\theta_i) \neq \mathbf{a}(\theta_i'), \quad \forall\ \theta_i \neq \theta_i' \in \left[-\frac{\pi}{2}, \frac{\pi}{2}\right].
\]

By using the method of contradiction, assuming
\begin{equation}
\label{st13}
L = \operatorname{lcm}\left( \frac{\lambda}{p_{l_1}}, \dots, \frac{\lambda}{p_{l_N}} \right) < 2.
\end{equation}

Then, spatial aliasing arising from periodic spatial under-sampling occurs.
There may exist distinct $\theta_i \neq \theta_i'$ such that for all $n = 1,\dots,N$ and some integers $k_{l_n} \in \mathbb{Z}$ \cite{ThompsonAR2017},
\[
2\pi \frac{p_{l_n}}{\lambda} \sin(\theta_i) = 2\pi \left( \frac{p_{l_n}}{\lambda} \sin(\theta_i') + k_{l_n} \right),
\]
or equivalently,
\[
\sin(\theta_i) \equiv \sin(\theta_i') \mod \left( \frac{\lambda}{p_{l_n}} \right).
\]

This indicates that two distinct angles $\theta_i \ne \theta_i'$ yield identical steering vectors, i.e., $\mathbf{a}(\theta_i) = \mathbf{a}(\theta_i')$,
thereby violating the injective mapping from angle $\theta$ to its steering vector $\mathbf{a}(\theta)$.
To eliminate angular ambiguity, the total unambiguous angular range-governed by the least common multiple (LCM) of the spatial sampling intervals must span at least one full period, i.e.,
$L \geq 2$, which contradicts the constraint in (\ref{st13}).
Therefore, the LCM condition is both necessary and sufficient for guaranteeing a unique mapping between the $\mathbf{x}(t)$  and $\mathbf{\theta}_{[1:D]}$
to reconstruct the source signals $\mathbf{s}(t)$.

\end{proof}

For LR-SDA with $N$-sensors located at the set $\mathbb{S}=\mathbb{S}_1\cup \mathbb{S}_2 \cup \mathbb{S}_3$,
the necessary and sufficient conditions for signal reconstruction are given by the following theorem.
\begin{theorem}[Necessary and Sufficient Condition for Source Signal Reconstruction]
Suppose $\{p_{l_1}, p_{l_2}, ..., p_{l_N} \in \mathbb{S}_1 \cup \mathbb{S}_2 \cup \mathbb{S}_3 \}$, where $\mathbb{S}_1$, $\mathbb{S}_2$, and $\mathbb{S}_3$ are defined as in (\ref{st1}).
Then, the source signals $\mathbf{s}(t)$ can be reconstructed from $\mathbf{x}(t)$ and $\mathbf{\theta}_{[1:D]}$ if and only if there exist coefficients $\{c_n\}_{n=1}^N$ and a positive integer $k$ such that
\begin{equation}\label{eq:coef_condition_new}
\begin{cases}
c_{n_1} = \dfrac{k \cdot p_{l_{n_1}}}{\mathrm{LCM}_1}, \quad n_1 = 1,2,...,N_1,\\
c_{n_2} = \dfrac{k \cdot p_{l_{n_2}}}{\mathrm{LCM}_2}, \quad n_2 = N_1+1, ..., N_1+N_2-\eta,\\
c_{n_3} = \dfrac{k \cdot p_{l_{n_3}}}{\mathrm{LCM}_3}, \quad n_3 = N_1+N_2-\eta+1, ..., N_1+N_2,\\
k \geq \dfrac{2\cdot \mathrm{lcm}(\mathrm{LCM}_1,\mathrm{LCM}_2,\mathrm{LCM}_3)}{\lambda},
\end{cases}
\end{equation}
where
\begin{align*}
\mathrm{LCM}_1 &= \mathrm{LCM}(\delta, N_2+1, N_1),\\
\mathrm{LCM}_2 &= \dfrac{((N_1-1)(N_2+1)+N_2)!}{((N_1-1)(N_2+1)+\eta+1)!},\\
\mathrm{LCM}_3 &= \dfrac{(N_1(N_2+1)+\eta+\delta)!}{(N_1(N_2+1)+1+\delta)!}.
\end{align*}
\end{theorem}

\begin{proof}
\textbf{Sufficiency:}
When there exist coefficients $\{c_n\}_{n=1}^N$ and an integer $k$ satisfying the conditions in (\ref{eq:coef_condition_new}).

We analyze the three subsets separately:

\emph{1) Analysis of $\mathbb{S}_1$:}
The elements $\{p_{l_1},...,p_{l_{N_1}}\}$ form an arithmetic sequence with initial term $\delta$ and common difference $N_2+1$.
By number theory results~\cite{Xiao2023}, their least common multiple is calculated by Algorithm 1
\[
\mathrm{LCM}_1 = \mathrm{LCM}(\delta, N_2+1, N_1).
\]
The assigned coefficients satisfy
\[
c_{n_1} = \dfrac{k \cdot p_{l_{n_1}}}{\mathrm{LCM}_1}.
\]

\emph{2) Analysis of $\mathbb{S}_2$:}
The elements $\{p_{l_{N_1+1}},...,p_{l_{N_1+N_2-\eta}}\}$ are structured as defined in (\ref{st2}), and their least common multiple is
\[
\mathrm{LCM}_2 = \dfrac{((N_1-1)(N_2+1)+N_2)!}{((N_1-1)(N_2+1)+\eta+1)!}.
\]
The corresponding coefficients satisfy
\[
c_{n_2} = \dfrac{k \cdot p_{l_{n_2}}}{\mathrm{LCM}_2}.
\]

\emph{3) Analysis of $\mathbb{S}_3$:}
Similarly, for the elements $\{p_{l_{N_1+N_2-\eta+1}},...,p_{l_{N_1+N_2}}\}$, we have
\[
\mathrm{LCM}_3 = \dfrac{(N_1(N_2+1)+\eta+\delta)!}{(N_1(N_2+1)+1+\delta)!},
\]
and the coefficients satisfy
\[
c_{n_3} = \dfrac{k \cdot p_{l_{n_3}}}{\mathrm{LCM}_3}.
\]

\emph{4) Calculation of Overall Least Common Multiple:}
The overall least common multiple of $\mathrm{LCM}_1$, $\mathrm{LCM}_2$ and $\mathrm{LCM}_3$ is computed as
\[
\mathrm{lcm}(\mathrm{LCM}_1,\mathrm{LCM}_2,\mathrm{LCM}_3)\sx \triangleq \sx \dfrac{\dfrac{\mathrm{LCM}_1\sx \cdot\sx \mathrm{LCM}_2}{\gcd(\mathrm{LCM}_1,\mathrm{LCM}_2)} \sx\cdot\sx \mathrm{LCM}_3}{\gcd( \dfrac{\mathrm{LCM}_1\sx\cdot\sx \mathrm{LCM}_2}{\gcd(\mathrm{LCM}_1,\sx \mathrm{LCM}_2)}, \mathrm{LCM}_3 )}.
\]

\emph{5) Guarantee of Source Reconstruction:}
The condition on $k$ ensures that
\[
k \geq \dfrac{2 \cdot \mathrm{lcm}(\mathrm{LCM}_1,\mathrm{LCM}_2,\mathrm{LCM}_3)}{\lambda}.
\]
Accordingly, the least common multiple
\[
\mathrm{lcm}\left(\dfrac{\lambda}{p_{l_1}}, \dfrac{\lambda}{p_{l_2}}, ..., \dfrac{\lambda}{p_{l_N}}\right) \geq 2,
\]
which satisfies the requirement in Theorem~3, thus guaranteeing the successful reconstruction of the source signals $\mathbf{s}(t)$.

\begin{table}[h]
\begin{center}
\label{tab1}
\renewcommand{\arraystretch}{0.9}
\begin{tabular}{ l }
\hline
\textbf{Algorithm 1: Calculate the least common multiple } \\
\textbf{ \ \ \ \ \ \ \ \ \ \ \ \ \ \ \ \ of sequence} \\
\hline
\bf{Input}: Set $\mathbb{S}$.\\
$\beta_1=\mathbb{S}[1]$, $\beta_2=\mathbb{S}[2]-\mathbb{S}[1]$, $N=length(\mathbb{S})$. \\
\textbf{Function LCM($\beta_1$,$\beta_2$):}\\
\ \ \ \ return ($|\beta_1 \cdot (\beta_1+\beta_2) /gcd(\beta_1,\beta_1+\beta_2)|$)\\
end\\
\textbf{Function LCM-multiple(numbers):}\\
\ \ \ \ Initialization: set result = numbers[0]\\
\ \ \ \ for each number in numbers[1:]:\\
\ \ \ \ \ \ \ \ result = lcm(result,number)\\
\ \ \ \ end\\
\ \ \ \ return result\\
end\\
\textbf{Function LCM-sequence($\beta_1$,$\beta_2$,$N$):}\\
\ \ \ \ Define sequence as empty list\\
\ \ \ \ for i from 0 to $N-1$:\\
\ \ \ \ \ \ \ \ add ($\beta_1+i\times \beta_2$) to sequence\\
\ \ \ \ end\\
\ \ \ \ return lcm-multiple(sequence) $\triangleq$ LCM($\beta_1$, $\beta_2$, $N$)\\
end\\
\textbf{Output}: LCM($\beta_1$, $\beta_2$, $N$) Least common multiple of $\mathbb{S}$.\\
\hline
\end{tabular}
\end{center}
\end{table}

\textbf{Necessity:}
Conversely, if the source signals $\mathbf{s}(t)$ can be reconstructed according to Theorem~3, then the least common multiple satisfies
\[
\mathrm{lcm}\left(\dfrac{\lambda}{p_{l_1}}, \dfrac{\lambda}{p_{l_2}}, ..., \dfrac{\lambda}{p_{l_N}}\right) \geq 2.
\]

Thus, it is necessary that $k$ be chosen sufficiently large according to
\[
k \geq \dfrac{2\cdot \mathrm{lcm}(\mathrm{LCM}_1,\mathrm{LCM}_2,\mathrm{LCM}_3)}{\lambda},
\]
with coefficients $\{c_n\}_{n=1}^N$ assigned accordingly as in (\ref{eq:coef_condition_new}).

Hence, the conditions in (\ref{eq:coef_condition_new}) are also necessary.


\end{proof}

\section{Weight Function}
As already mentioned in Section III, the sensor positions in the co-arrays represented by the multisets $\Phi_1$, $\Phi_2$, $\Phi_3$ and $\Phi_4$,
may be repetitive. Therefore, the output corresponding to each sensor located at a particular
position is obtained by averaging the outputs of all sensor positions.
Definitely, this averaging makes the estimate of the corresponding sensor better by decreasing its variance \cite{Dias2017}
and makes the output more stable.
Formally, the stability of the designed array is studied using the weight function.
Therefore, the weight function plays an important
role in designing the array. The weight function of any array $\Phi$ is defined as a function that maps every element
$z \in \Phi^u$ to a positive integer giving the multiplicity of the element $z$ in the multiset $\Phi$.
Therefore, the weight function of the SO-ECA $\Phi^u$ associated with the proposed LR-SDA, is denoted as $W_{\Phi^u}^{\mathbb{S}}$,
and defined as follows\\[-2pt]
\begin{equation}\nonumber
\begin{aligned}
W_{z\in \Phi^u}^{\mathbb{S}^2}\triangleq
\left|
\left\{                 
  \begin{array}{c|c}   
                      &  \{ p_{l_1}, p_{l_2}\}\in \mathbb{S}\\
  \{ p_{l_1}, p_{l_2}\} & z\in \left\{ \begin{array}{c}  p_{l_1}+p_{l_2} \\ p_{l_1}-p_{l_2} \\ -p_{l_1}+p_{l_2} \\  -p_{l_1}-p_{l_2} \end{array}\right\}
  \end{array}
\right\}
\right|.
\end{aligned}
\end{equation}

The proposed LR-SDA consists of three arrays with sensors located at sets $\mathbb{S}_1$, $\mathbb{S}_2$ and $\mathbb{S}_3$.
To lighten the notations, the $\mathbb{S}_2\cup\mathbb{S}_3\triangleq \mathbb{S}_{23}$ is considered as one sub-array.
Hence the weight function $W_{z\in \Phi^u}^{\mathbb{S}}$ can be calculated as follows
\begin{equation}
\label{wto3}
W_{z\in \Phi^u}^{\mathbb{S}^2}=W_{z\in \Phi^u}^{\mathbb{S}_1^2}+W_{z\in \Phi^u}^{\mathbb{S}_1\mathbb{S}_{23}}+W_{z\in \Phi^u}^{\mathbb{S}_{23}^2},
\end{equation}
where $W_{z\in \Phi^u}^{\mathbb{S}_1^2}$ and $W_{z\in \Phi^u}^{\mathbb{S}_{23}^2}$ are considered as $\{ p_{l_1}, p_{l_2}\}\in\mathbb{S}_1^2 $
and $\{ p_{l_1}, p_{l_2}\}\in\mathbb{S}_{23}^2 $, respectively. So termed as weight functions of types \textit{'inter-ULA-I'} and \textit{'inter-ULA-II'}, respectively.
Whereas, $W_{z\in \Phi^u}^{\mathbb{S}_1\mathbb{S}_{23}}$ is considered as $\{ p_{l_1}, p_{l_2}\}\in\mathbb{S}_1\mathbb{S}_{23} $, so termed as weight function of type
\textit{'inter-ULA-12'}. Moreover, the weight function of type \textit{'inter-ULA-12'} is calculated by considering all possible combinations, i.e.
\begin{equation}
W_{z\in \Phi^u}^{\mathbb{S}_1\mathbb{S}_{23}}=W_{z\in \Phi^u}^{\{p_{l_1}\in\mathbb{S}_1,p_{l_2}\in\mathbb{S}_{23}\}}
+W_{z\in \Phi^u}^{\{p_{l_2}\in\mathbb{S}_1,p_{l_1}\in\mathbb{S}_{23}\}},
\end{equation}
Since SO-ECA $\Phi$ is the union of $\Phi_1$, $\Phi_2$, $\Phi_3$ and $\Phi_4$,
the weight function of LR-SDA $\mathbb{S}$ can be expressed as the sum of the above-mentioned three co-arrays as follows
\begin{equation}
\label{wto4}
\begin{aligned}
&W_{z\in \Phi^u}^{\mathbb{S}_1^2}=W_{z\in \Phi_1^u}^{\mathbb{S}_1^2}+W_{z\in \Phi_2^u}^{\mathbb{S}_1^2}+W_{z\in \Phi_3^u}^{\mathbb{S}_1^2}+W_{z\in \Phi_4^u}^{\mathbb{S}_1^2},\\
&W_{z\in \Phi^u}^{\mathbb{S}_{23}^2}=W_{z\in \Phi_1^u}^{\mathbb{S}_{23}^2}+W_{z\in \Phi_2^u}^{\mathbb{S}_{23}^2}+W_{z\in \Phi_3^u}^{\mathbb{S}_{23}^2}+W_{z\in \Phi_4^u}^{\mathbb{S}_{23}^2},\\
&W_{z\in \Phi^u}^{\mathbb{S}_1\mathbb{S}_{23}}=W_{z\in \Phi_1^u}^{\{p_{l_1}\in\mathbb{S}_1,p_{l_2}\in\mathbb{S}_{23}\}}
+W_{z\in \Phi_1^u}^{\{p_{l_2}\in\mathbb{S}_1,p_{l_1}\in\mathbb{S}_{23}\}}\\
&\ \ \ \ \ \ \ \ \ \ +W_{z\in \Phi_2^u}^{\{p_{l_1}\in\mathbb{S}_1,p_{l_2}\in\mathbb{S}_{23}\}}
+W_{z\in \Phi_2^u}^{\{p_{l_2}\in\mathbb{S}_1,p_{l_1}\in\mathbb{S}_{23}\}}\\
&\ \ \ \ \ \ \ \ \ \ +W_{z\in \Phi_3^u}^{\{p_{l_1}\in\mathbb{S}_1,p_{l_2}\in\mathbb{S}_{23}\}}
+W_{z\in \Phi_3^u}^{\{p_{l_2}\in\mathbb{S}_1,p_{l_1}\in\mathbb{S}_{23}\}}\\
&\ \ \ \ \ \ \ \ \ \ +W_{z\in \Phi_4^u}^{\{p_{l_1}\in\mathbb{S}_1,p_{l_2}\in\mathbb{S}_{23}\}}
+W_{z\in \Phi_4^u}^{\{p_{l_2}\in\mathbb{S}_1,p_{l_1}\in\mathbb{S}_{23}\}}.
\end{aligned}
\end{equation}


Substituting (\ref{wto4}) into (\ref{wto3}), the results in the overall weight function $W_{z\in \Phi^u}^{\mathbb{S}^2}$ can be obtained.

Since NADiS, TNA-I, TNA-II and the proposed LR-SDA are all designed based on SCA and DCA,
comparing the weight functions of them with $N=9$ physical sensors in Fig. 3.
It can be seen that the consecutive lags of LR-SDA more than those of other three arrays.
In addition, the number of each lag is shown in the Fig. 3, which show the contributions to the stability of arrays.

\begin{figure}[h]
  \centering
  \subfigure[]{
    \label{fig:subfig:onefunction}
    \includegraphics[scale=0.3]{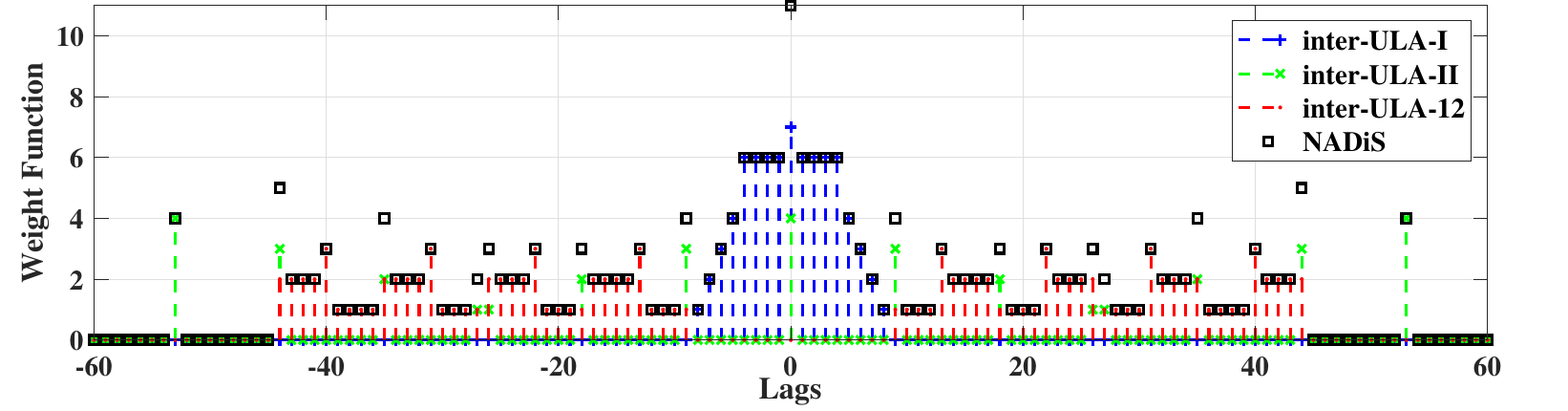}}
  \hspace{0in} 
  \subfigure[]{
    \label{fig:subfig:threefunction}
    \includegraphics[scale=0.3]{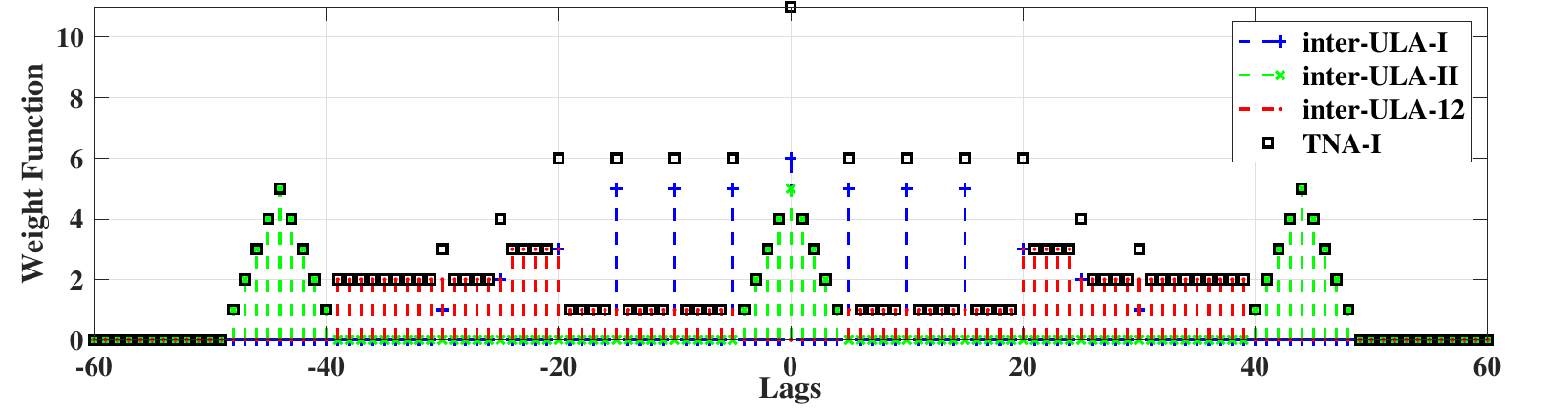}}
  \subfigure[]{
    \label{fig:subfig:threefunction}
    \includegraphics[scale=0.3]{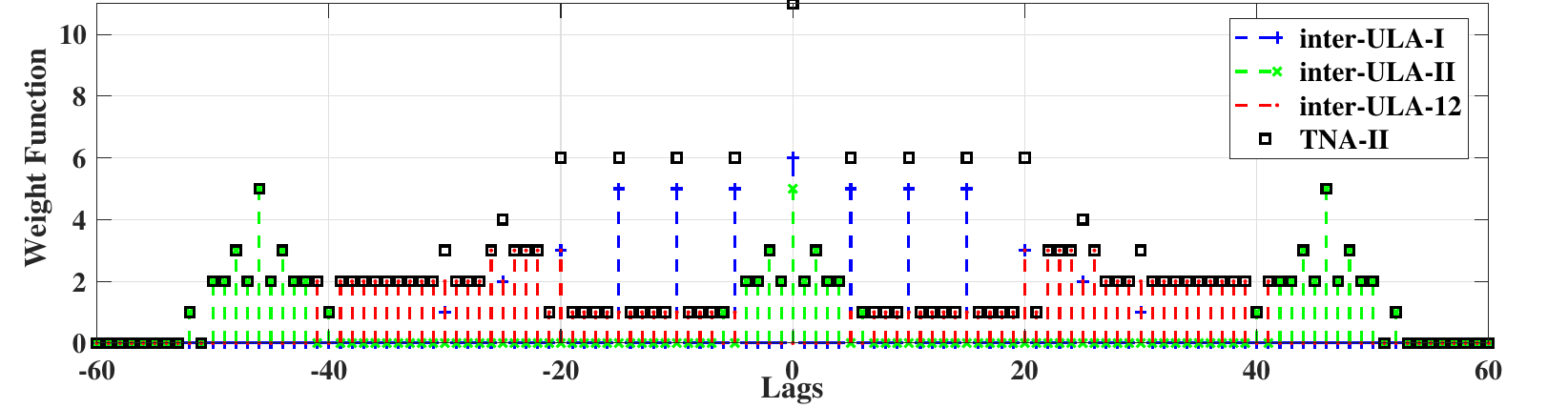}}
  \subfigure[]{
    \label{fig:subfig:threefunction}
    \includegraphics[scale=0.3]{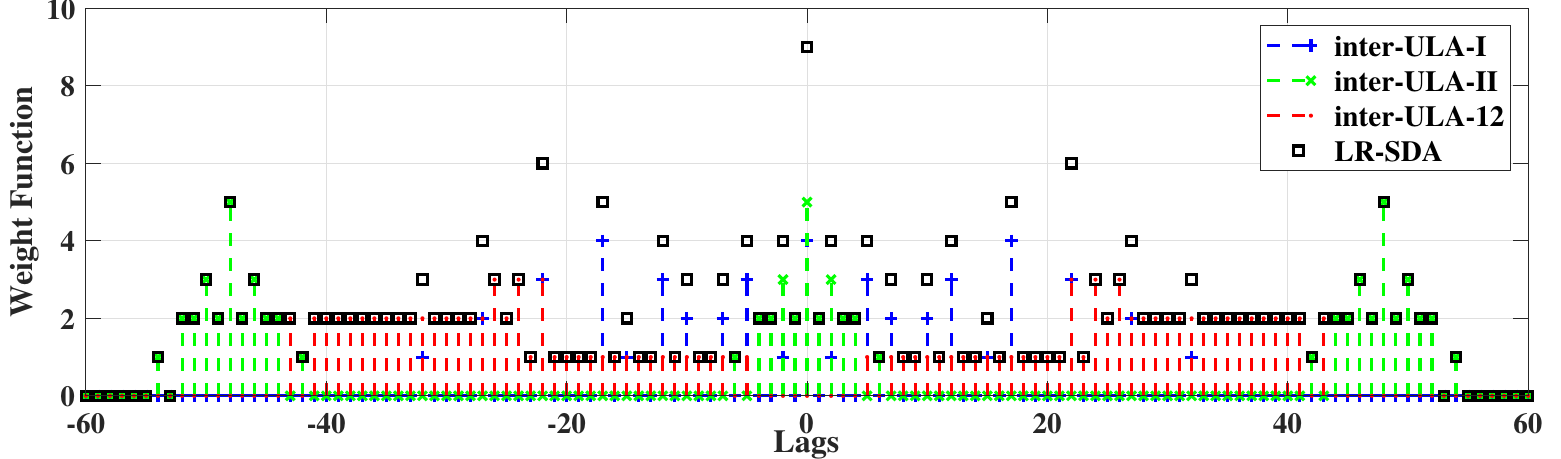}}
  \caption{Synthesis of weight function of the SO-ECA associated with four arrays when the number of physical sensors $N=9$.
   (a) NADiS. (b) TNA-I. (c) TNA-II. (d) LR-SDA. }
\end{figure}

\section{Redundancy}
Redundancy quantifies the discrepancies between the number of sensors in consecutive segment and total number of sensors for the virtual array of an array.
Further, in order to define the redundancy of LR-SDA, the redundancies of DCA and SCA are introduced as follows.
\begin{definition}
(Redundancy\cite{Hoctor1990}). The redundancy of a $N$-sensors SLA with a consecutive sum co-array is
\begin{equation}
R_S=\frac{N(N+1)/2}{2N+1}\geq 1.
\end{equation}
\end{definition}

\begin{definition}
(Redundancy \cite{Moffet1968}). The redundancy of a $N$-sensors SLA with a consecutive difference co-array is
\begin{equation}
R_D=\frac{N(N-1)/2}{L}\geq 1,
\end{equation}
where $L$ is the aperture of the consecutive difference co-array.
\end{definition}

\subsection{Redundancy of SO-ECA}

In order to define the redundancy of SO-ECA, the size of $\mathbb{U}$ is studied.
According to the redundancy definitions of DCA and SCA, it can be seen that $2N+1$ and $2L+1$ are the consecutive segment in SCA and DCA.
Further, it is can be known that the bounds of $R_S$ and $R_D$ are characterized the size of $N$ and $L$.
Therefore, the SO-ECA redundancy is defined as follows by generalizing the redundancy definition idea of SCA and DCA.
\begin{definition}
The redundancy of SO-ECA is defined as
\begin{equation}
R_{w}=\frac{\tilde{k}(N)}{U},
\end{equation}
where $\tilde{k}(N)=(N(N-1)/2+N(N+1)/2)=N^2$ denotes the one-side maximal size of $\Phi$ with $N$-sensors and $|\mathbb{U}|$ is the consecutive segment of $\Phi^u$,
which means $\mathbb{U}(\mathbb{S},\mathbb{S}')=[-U:U]\subseteq\Phi^u$.
\end{definition}

The redundancy of SO-ECA quantifies the discrepancies between the consecutive segment $\mathbb{U}$ and the maximum number of sensors obtained from SO-ECA.
If $R_{w} = 1$, all the second-order differences and sums are distinct with a ULA.
If $R_{w} >1$, either the normalized size $\Phi^u$ is smaller than 1, or there are holes within its $\Phi^u$.

The definition of $R_{w}$ uses two one-sided quantities $\tilde{k}(N)$ and $U$ instead of two-sided quantities $k(N)$ and $\mathbb{U}$.
This definition follows the convention of the second-order redundancy in \cite{Moffet1968},
which can be further traced back to the topic of difference basis in number theory \cite{Redei1948,Erdos1948,Leech1956}.
Next, the bounds of $R_{w}$ is discussed. Since $U$ may be $\{0\}$, in which case $R_{w} = \infty$,
we skip the discussion on the upper bound of $R_{w}$.

\begin{theorem}
The redundancy $R_{w}$ of  SO-ECA satisfies
\begin{equation}
\label{wto12}
\begin{aligned}
R_{w}>L_2(N):&=(1+\frac{2}{3\pi})\frac{\tilde{k}(N)}{\left( \left(\begin{matrix} N \\ 2 \\ \end{matrix} \right) \right)}
=(1+\frac{2}{3\pi})\frac{2N^2}{N^2+N}.
\end{aligned}
\end{equation}
\end{theorem}

\begin{proof}
Inspired by the proof method in \cite{Redei1948} and \cite{Linel1993}, the lower bound of $R_{w}$ is given by a function of $y$
\begin{equation}\nonumber
\begin{aligned}
f(y)&=\sum_{\substack{0\leq n_1 \leq N-1 \\ 0\leq n_2 \leq N-1}}[\cos((l_{n_1}+l_{n_2})y)+\cos((l_{n_1}-l_{n_2})y)]\\
\end{aligned}
\end{equation}

By the definition of $\mathbb{U}$, for each integer $u \in [-U, U]$, we select one 2-tuple
$T_u= (n_1(u), n_2(u))$ such that $n_1(u) + n_2(u)= u_1$, $n_1(u) - n_2(u)= u_2$, $0 \leq n_1(u) \leq N-1$,
and $0 \leq n_2(u) \leq N_1$. Denoting $\mathcal{T} :=\{T_u | u \in [-U, U]\}$, we can decompose $f(y)$ into

\begin{equation}\nonumber
\begin{aligned}
f(y)&=\sum_{( n_1, n_2) \in \mathcal{T}}[\cos((l_{n_1}+l_{n_2})y)+\cos((l_{n_1}-l_{n_2})y)]\\
&+\sum_{( n_1, n_2) \notin \mathcal{T}}[\cos((l_{n_1}+l_{n_2})y)+\cos((l_{n_1}-l_{n_2})y)].
\end{aligned}
\end{equation}

Since $\cos y\leq1$ for all $y\in \mathbb{R}$ and $|\mathcal{T}|=2U+1$, we can get
\begin{equation}\nonumber
\begin{aligned}
f(y)&\leq\sum_{u_1=-U}^{U}\cos(u_1y)+\sum_{u_2=-U}^{U}\cos(u_2y)
+2\left( \left(
               \begin{matrix}
                 N \\
                 2                    \\
               \end{matrix}
         \right) \right)
             -|\mathcal{T}|.
\end{aligned}
\end{equation}

To further discuss the value of $f(y)$, $\sum\limits_{u=-U}^{U}\cos(uy)$ is considered as follows.
Since $\cos(uy)=\frac{e^{iuy}+e^{-iuy}}{2}$, we can get
\begin{equation}
\label{wto6}
\begin{aligned}
\sum_{u=-U}^{U}\cos(uy)=\frac{1}{2}(\sum_{u=-U}^{U}e^{iuy}+\sum_{u=-U}^{U}e^{-iuy}),
\end{aligned}
\end{equation}
for the positive of exponent in (\ref{wto6}), it can be rewritten as follows
\begin{equation}
\label{wto7}
\begin{aligned}
\Upsilon&\triangleq\sum_{u=-U}^{U}e^{iuy}=\sum_{u=0}^{U}e^{iuy}+\sum_{u=-U}^{-1}e^{iuy}.
\end{aligned}
\end{equation}

According to the symmetry of $\Upsilon$, we can get
\begin{equation}\nonumber
\begin{aligned}
\sum_{u=1}^{U}e^{-iuy}=\sum_{u=-U}^{-1}e^{iuy}.
\end{aligned}
\end{equation}

Therefore, $\Upsilon$ in (\ref{wto7}) can be transformed as
\begin{equation}\nonumber
\begin{aligned}
\Upsilon&=1+\sum_{u=1}^{U}(e^{iuy}+e^{-iuy})
=1+\sum_{u=1}^{U}2\cos(uy).
\end{aligned}
\end{equation}

In addition, for $\sum_{u=-U}^{U}e^{iux}=e^{-iUy}\sum_{k=0}^{2U}e^{iky}$,  we can get the following equation from the principle of summing geometric series
\begin{equation}
\label{wto8}
\begin{aligned}
\sum_{k=0}^{K}q^k=\frac{1-q^{N+1}}{1-q},\ \ (q=e^{iy}),
\end{aligned}
\end{equation}
substituting (\ref{wto8}) into (\ref{wto7}), we can obtain
\begin{equation}\nonumber
\begin{aligned}
\sum_{u=-U}^{U}e^{iuy}=e^{-iUy}\cdot \frac{1-e^{i(2U+1)y}}{1-e^{iy}}.
\end{aligned}
\end{equation}

Since $1-e^{iy}=1-\cos y-i\sin y$, the modulo of it is as follows
\begin{equation}
\label{wto9}
\begin{aligned}
|1-e^{iy}|=2\sin\frac{y}{2}.
\end{aligned}
\end{equation}

Similarly, for $1-e^{i(2U+1)y}=1-\cos((2U+1)y)-i\sin((2U+1)y)$, the modulo of it is as follows
\begin{equation}
\label{wto10}
\begin{aligned}
|1-e^{i(2U+1)y}|=2\sin[(U+\frac{1}{2})y].
\end{aligned}
\end{equation}

Therefore, substituting (\ref{wto9}) and (\ref{wto10}) into (\ref{wto6})  we can get
\begin{equation}\nonumber
\begin{aligned}
\sum_{u=-U}^{U}\cos(uy)=\frac{\sin[(U+\frac{1}{2})y]}{\sin\frac{y}{2}}.
\end{aligned}
\end{equation}

We can get the following inequality due to the non negativity of $f(y)$ for all $y$
\begin{equation}\nonumber
\begin{aligned}
\frac{2\sin[(U+\frac{1}{2})y]}{\sin\frac{y}{2}}+
2\left( \left(
               \begin{matrix}
                 N \\
                 2                    \\
               \end{matrix}
         \right) \right)
-2U-1 \geq f(y)\geq0.
\end{aligned}
\end{equation}

By rearranging this inequality and selecting $x=\frac{3\pi}{2U+1}$ for $U \geq 1$, we can obtain
\begin{equation}
\label{wto11}
\begin{aligned}
\frac{\left( \left(
               \begin{matrix}
                 N \\
                 2                    \\
               \end{matrix}
         \right) \right)}{U} \geq (1 - \frac{\sin[(U_2+\frac{1}{2})y]}{U\sin\frac{y}{2}}) + \frac{1}{2U} > (1+\frac{2}{3\pi}).
\end{aligned}
\end{equation}

The desired inequality can be obtained by multiplying (\ref{wto11}) with $\tilde{k}(N)/ \left( \left(\begin{matrix} N \\ 2 \\ \end{matrix} \right) \right)$ on both sides.

\end{proof}

The lower bound $L_2(N)$ leads to insights into the SO-ECA. Among all sparse arrays, the size of $\mathbb{U}$ is strictly smaller
than $k(N)$ for $N \geq 2$. The reason is as follows.
Since it can be shown that $L_2(N)$ is an increasing function of $N$,
the redundancy of SO-ECA satisfies $R_{w} > L_2(N) \geq L_2(1) \approx 1.2122$ for $N \geq 1$.
This relation indicates that for $N \geq 2$, $\tilde{k}(N)>U$, or equivalently $|\mathbb{U}| < k(N)$.

For sufficiently large $N$, we can derive a upper bound for $|\mathbb{U}|$,
which is stronger than $k(N)$. Based on (\ref{wto12}), $\tilde{k}(N)/L_2(N)$ is
approximately $0.4125N^2$ in this region. As a result, we have the following asymptotic relation for $|\mathbb{U}|$

\begin{equation}
|\mathbb{U}|=2U+1\leq 1+2\frac{\tilde{k}(N)}{L_2(N)}\approx 0.8249 N^2.
\end{equation}

\subsection{Redundancy of LR-SDA}
For the special case of SO-ECA, the $R_{w}$ of LR-SDA is considered.
We can get the $\mathbb{U}=2N_1(N_2+1)+2(\lceil \frac{N_2}{2} \rceil-1)+2\delta$ according to (\ref{st1}).
Therefore, the $R_{w}$ of the LR-SDA can be obtained based on the Definition 8 as follows
\begin{equation}
\begin{aligned}
&R_{w}^1=\frac{N^2}{2N_1(N_2+1)+2(\lceil \frac{N_2}{2} \rceil-1)+2\delta},\\
&\delta = \lfloor \frac{N_2 + 1}{2} \rfloor \ and \  (N_1 - 1)(N_2 + 1) + N_2+ \delta + 1 \in \mathbb{C}(\mathbb{S}_1,\mathbb{S}_1).
\end{aligned}
\end{equation}

\begin{corollary}
The upper and lower bounds of the redundancy $R_w^1$ with the number of physical sensors varying from 2 to infinity are
\begin{equation}\nonumber
\begin{aligned}
1\leq R_w^1\leq 2.
\end{aligned}
\end{equation}
\end{corollary}

\begin{proof}
See Appendix B.
\end{proof}

\section{Performance Comparison}
In this section, we provide numerical simulations to demonstrate the superior performance of
LR-SDA in terms of DOF, resolution and the RMSE versus the input SNR, snapshots and the number of sources.
Note that in all DOA estimations, the spatial smoothing MUSIC algorithm \cite{Pal22011}, \cite{Liu2015}, \cite{Piya2012},
\cite{You2021} is used to estimate DOA.
Moreover, we assume that all incident sources have equal power and the number of sources is known.
To evaluate the results quantitatively, the root-mean-square error (RMSE) of the estimation DOAs is defined as an
average over 1000 independent trials:

\begin{equation}
\text{RMSE}=\sqrt{\frac{1}{1000D}\sum_{j=1}^{1000}\sum_{i=1}^{D}(\hat{\theta}_i^{j}-\theta_i)^2},
\end{equation}
where $\hat{\theta}_i^{j}$ is the estimate of $\theta_i$ for the $j^{th}$ trial. Similar to \cite{LiuCL2016},
we focus on the DOF obtained by different arrays,
rather than the array aperture, to investigate the overall estimation performance.

\subsection{Comparison of the DOF for Different Arrays}

We compare the DOF of the proposed method with those of NADiS \cite{GuptaP2018},
TNA-I, TNA-II \cite{WangY2020}, TS-ENA \cite{Yang2023} and GENAMS for given the fixed number of physical sensors,
where all other second-order DCAs adopt the  array structure for obtaining maximum DOF.
The comparing results are listed in Table II and the variations of DOF for six methods are shown in Fig. 3.
The results in Table II show that the DOF of LR-SDA have increased compared to other five methods under the different given number of physical sensors.
In addition, it can be seen that as the number of sensors $N$ increases, the DOF of the proposed LR-SDA is always more than
those of the other five DCAs in Fig. 4, and the advantages of the proposed array are more obvious with a large number of physical sensors.

\begin{figure}[h]
 \center{\includegraphics[width=6cm]  {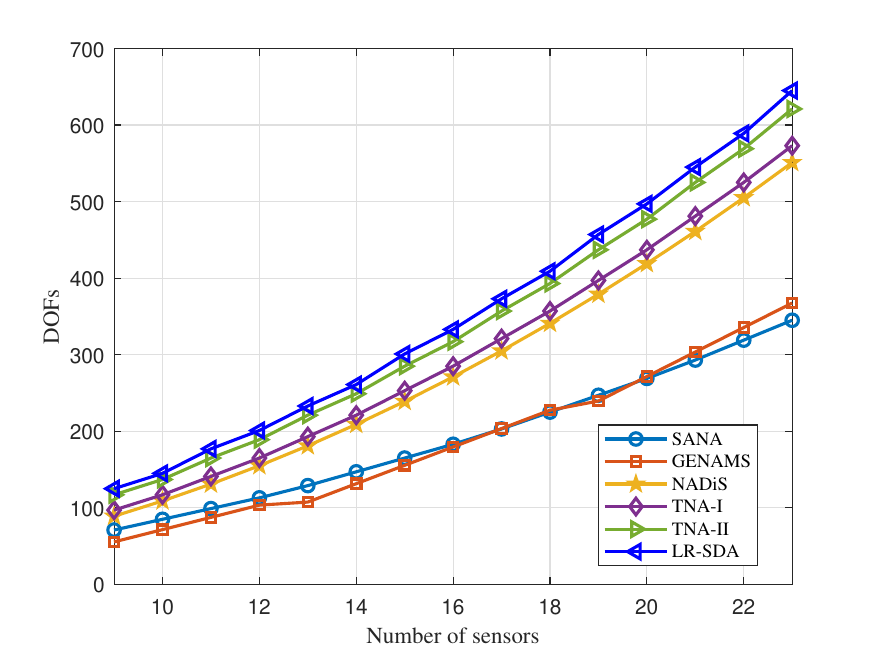}}
 \caption{\label{1} DOF of different arrays}
\end{figure}

\begin{table}[h]
\begin{center}
\caption{COMPARISON OF THE DOF FOR DIFFERENT ARRAYS
BASED ON SECOND-ORDER CUMULANT.}
\label{tab1}
\renewcommand{\arraystretch}{0.9}
\setlength{\tabcolsep}{7pt}
\begin{tabular}{ c  c  c  c }
\hline
\hline
\textbf{Array} ~~~~~~& \textbf{Hole-Free} & \textbf{Number of} &~~~~ \textbf{DOF}\\
\textbf{Structure}~~~~~~& \textbf{Co-Array}&\textbf{Sensors} &~~~~ \\
\hline
TS-ENA & Yes & $N$ & ${\mathcal{D}_1}^a$\\
GENAMS & Yes & $N$ & ${\mathcal{D}_2}^b$ \\
NADiS & Yes & $N_1+N_2$ & ${\mathcal{D}_3}^c$ \\
TNA-I & Yes &$N_1+N_2$  &${\mathcal{D}_4}^d$ \\
TNA-II & Yes & $N_1+N_2$ & ${\mathcal{D}_5}^e$\\
LR-SDA & Yes & $N_1+N_2$ & ${\mathcal{D}_6}^f$\\
\hline
\hline
\textbf{Array}~~~~~~ & \textbf{P or} & \textbf{Number of} & \textbf{DOF}\\
\textbf{Structure}~~~~~~ & \textbf{($N_1$,$N_2$)}& \textbf{Sensors}& \\
\hline
TS-ENA & (6,2) & 9 & 71\\
GENAMS & 7 & 9 & 59\\
NADiS & (5,4) & 9 & 89\\
TNA-I & (5,4) & 9 & 97\\
TNA-II & (5,4) & 9 & 101\\
LR-SDA & (5,4) & 9 & \textbf{109}\\
\hline
TS-ENA & (10,8) & 19 & 247\\
GENAMS & 15 & 19 & 247\\
NADiS & (10,9) & 19 & 379\\
TNA-I & (10,9) & 19 & 397\\
TNA-II & (10,9) & 19 & 415\\
LR-SDA & (10,9) & 19 & \textbf{425}\\
\hline
TS-ENA & (14,13) & 28 & 489\\
GENAMS & 19 & 28 & 541\\
NADiS & (14,14) & 28 & 811\\
TNA-I & (14,14) & 28 & 837\\
TNA-II & (14,14) & 28 & 865\\
LR-SDA & (14,14) & 28 & \textbf{879}\\
\hline
\hline
\end{tabular}
\end{center}
\footnotesize{$^a$ $2(N_1+1)N_2+7N_1+1$}\par
\footnotesize{$^b$ $2(4\frac{{P_3}+1}{4}+2\frac{P_3-1}{2}\frac{P_3+1}{4}+P_3(N-P_3)-2-\frac{P_3-1}{2}), P_3=4\lfloor\frac{N+5}{6}\rfloor$}\par
\footnotesize{$^c$ $N^2+N-1$}\par
\footnotesize{$^d$ $4N_1N_2+4N_1-3)$}\par
\footnotesize{$^e$ $4N_1N_2+4N_1+2N_2-3$ }\par
\footnotesize{$^f$ $-4N_2^2+(4N-4)N_2+4N+4\delta+1$}\\
\end{table}

\subsection{Redundancy of Different Array Structures}

The redundancy is an important indicator to measure whether the DOF of the current array structure can be further enhanced,
which is compared among NADiS, TNA-I, TNA-II, TS-ENA, GENAMS and LR-SDA in the simulation.
Redundancies of six arrays with the varying of the number of sensors are shown on the Fig. 5.
It can be seen that the redundancies of GENAMS and TS-ENA designed based on DCA is the smaller
than those of arrays designed based on SCA and DCA.
In addition, the redundancy of proposed LR-SDA is the smallest than those of other arrays designed based on SCA and DCA.

\begin{figure}
 \center{\includegraphics[width=6cm]  {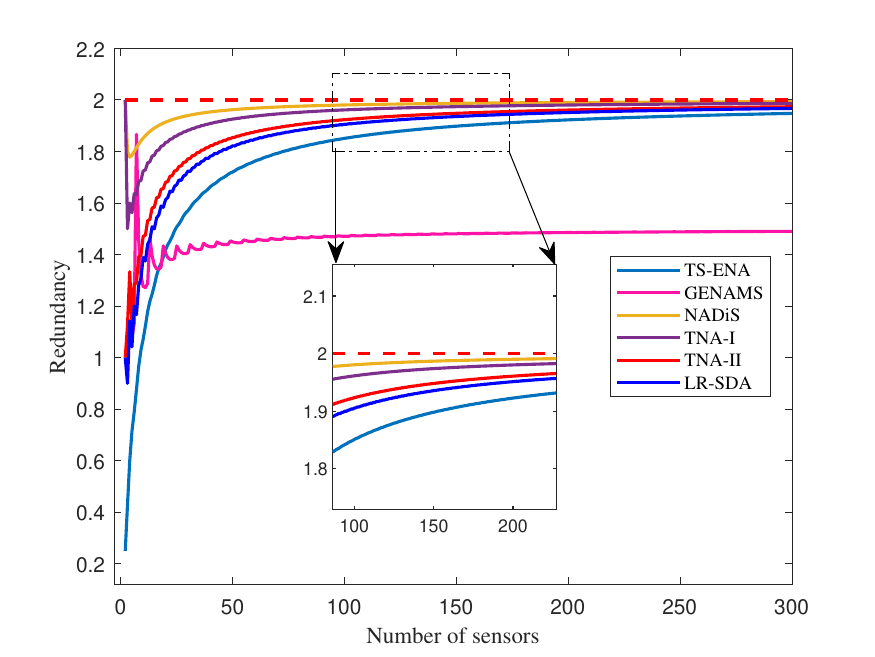}}
 \caption{\label{1} Redundancy of different arrays}
\end{figure}

\subsection{DOA Estimation Based on Second-Order Cumulants}
We compare the DOA estimation performance versus input SNR, snapshots and the number of sources for LR-SDA to
those of other five DCAs in this part,
where 11 physical sensors are used to construct six co-arrays.

Firstly, there are 20 uncorrelated sources uniformly located at $-60^{\circ}$ to $60^{\circ}$,
and the SNR and snapshots are set as 0 dB and 6000, respectively.
The DOA estimation results are shown in Fig. 6, where six arrays are capable of resolving all 20 sources.
However, the LR-SDA exhibits lowest valley near both ends than those of other five DCAs which can improve the performance of DOA estimation.

\begin{figure}
  \centering
  \subfigure[]{
    \label{fig:subfig:onefunction}
    \includegraphics[scale=0.17]{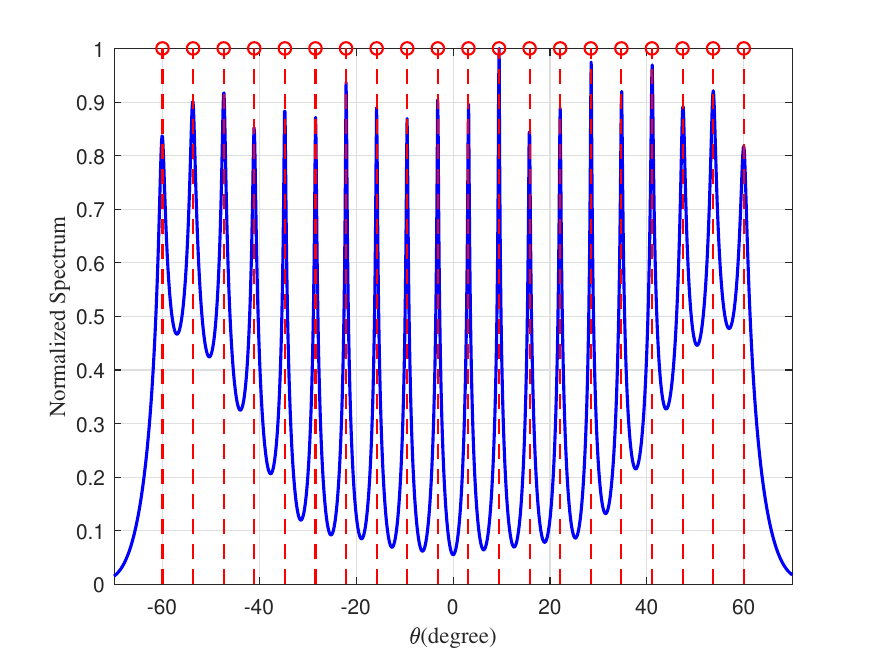}}
  \hspace{0in} 
  \subfigure[]{
    \label{fig:subfig:threefunction}
    \includegraphics[scale=0.17]{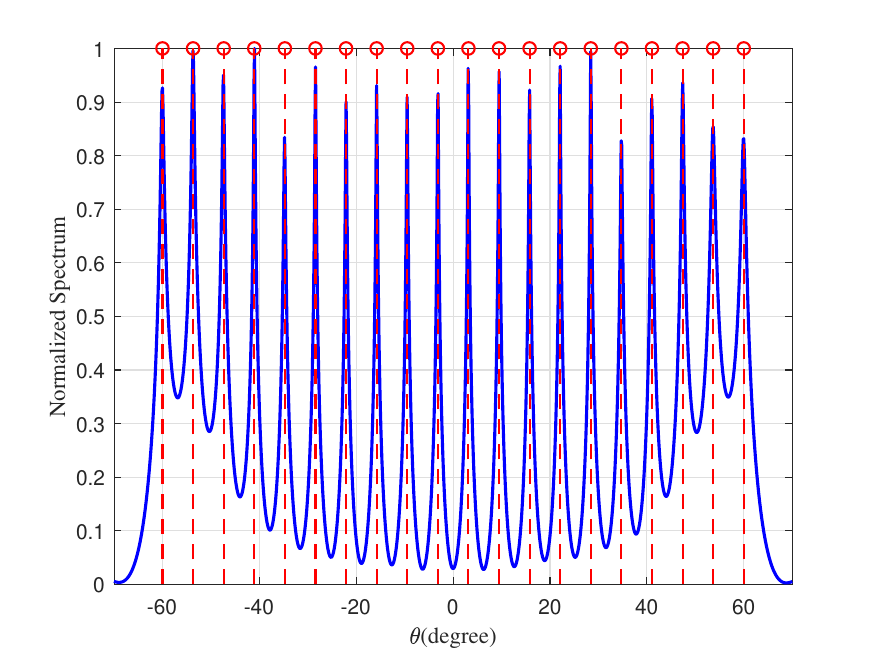}}
  \subfigure[]{
    \label{fig:subfig:threefunction}
    \includegraphics[scale=0.17]{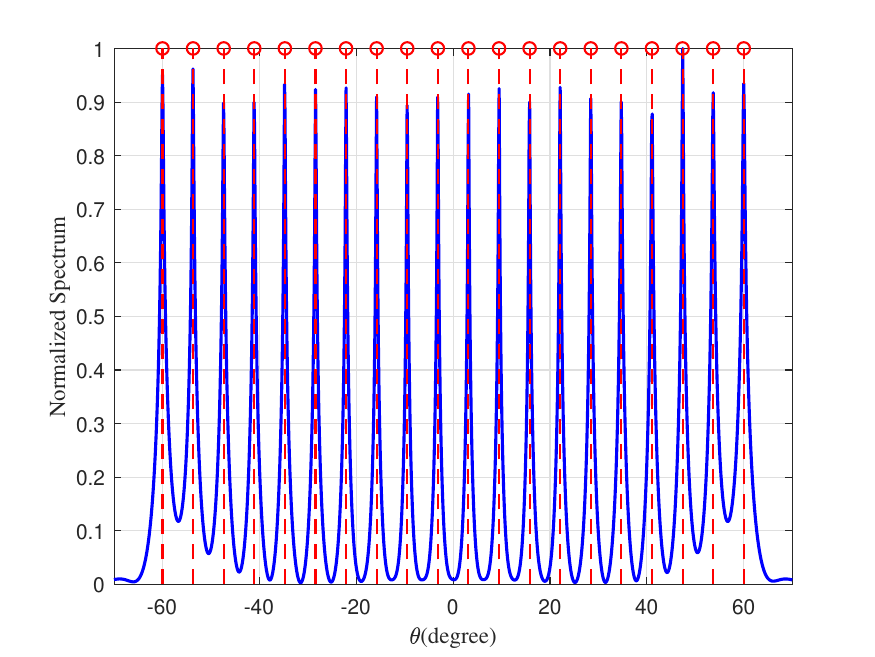}}
  \subfigure[]{
    \label{fig:subfig:threefunction}
    \includegraphics[scale=0.17]{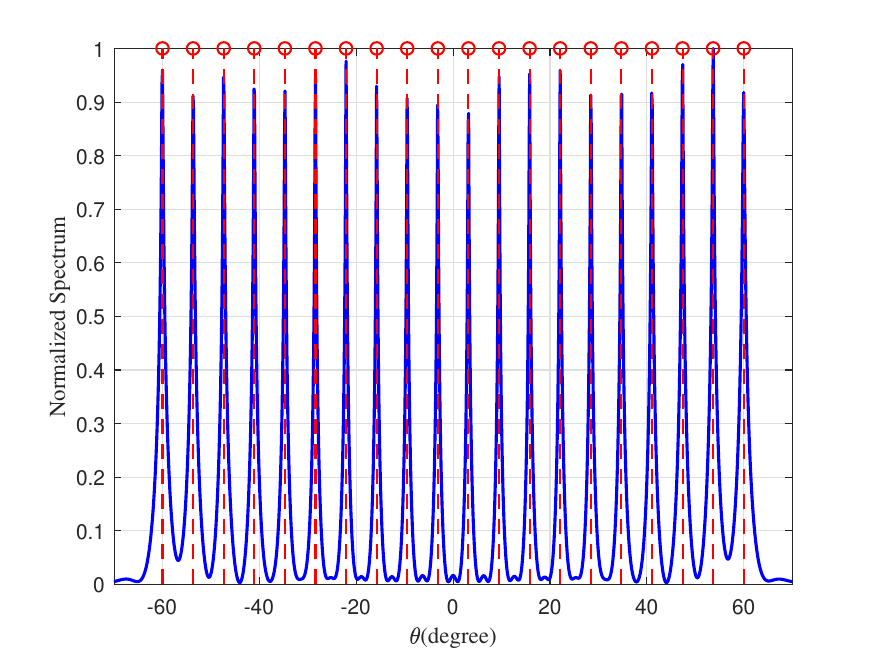}}
  \subfigure[]{
    \label{fig:subfig:threefunction}
    \includegraphics[scale=0.17]{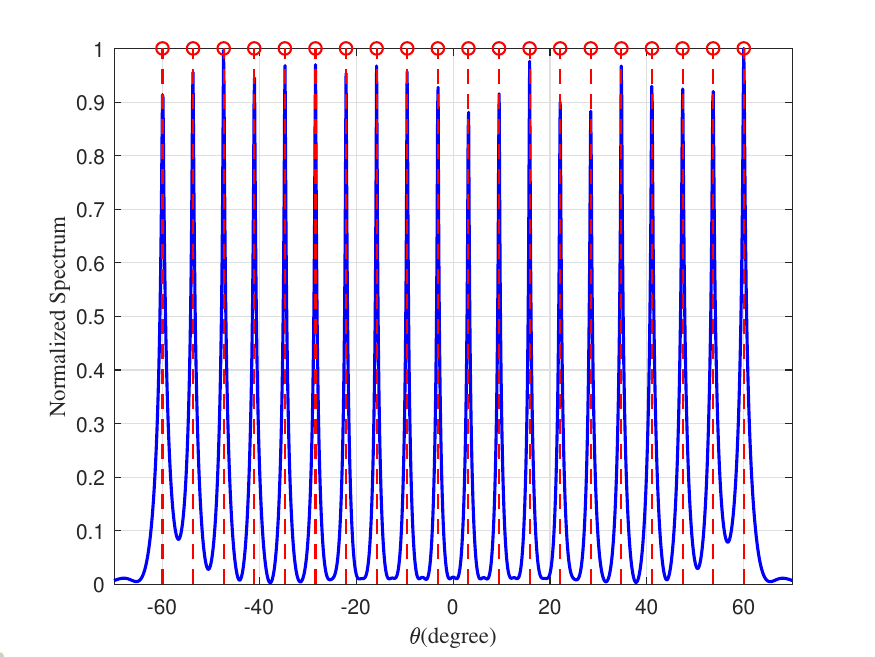}}
  \subfigure[]{
    \label{fig:subfig:threefunction}
    \includegraphics[scale=0.16]{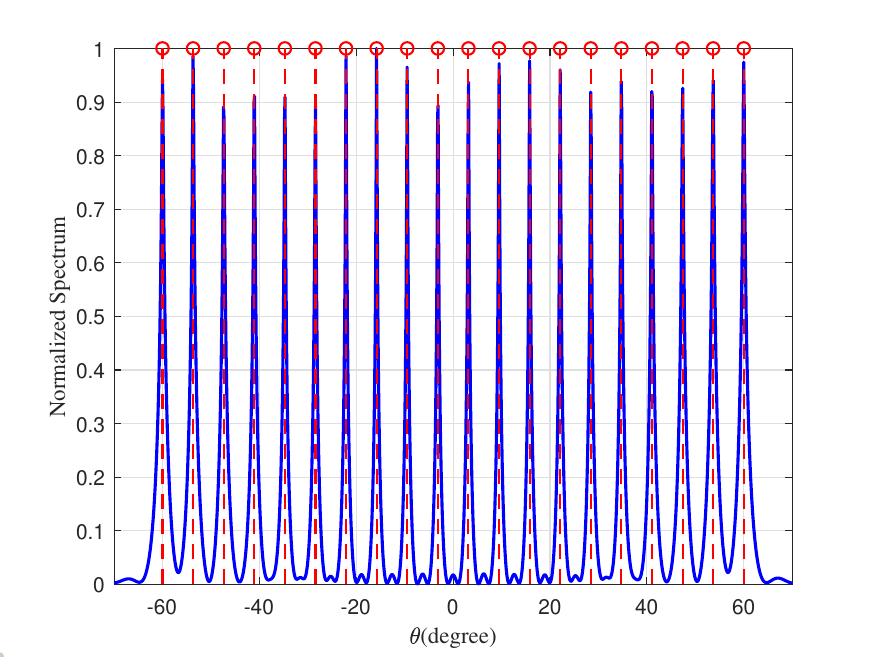}}
  \caption{DOA estimation result for five arrays with 11-sensors when 20 sources are uniformly located at $-60^0$ to $60^0$.
$SNR = 0$ dB and $K = 10000$. (a) GENAMS. (b) TS-ENA. (c) NADiS. (d) TNA-I. (e) TNA-II. (f) LR-SDA.}
\end{figure}

Secondly, the RMSE versus the input SNR, snapshots and the number of sources are studied in the following numerical simulations.
In the first numerical simulation, there are 12 uncorrelated sources uniformly located at $-60^{\circ}$ to $60^{\circ}$
and the snapshots setting as 12000. The SNR ranges from -10dB to 10dB with an interval of 2dB.
The results of RMSE versus SNR for different arrays are shown in Fig. 7(a),
where it can be seen that as the SNR increases, the RMSEs of all arrays decrease,
however the RMSE of LR-SDA remaining the lowest.
The second numerical simulation studies the DOA estimation performance with respect to the snapshots changing from 8000 to 18000
and the SNR setting as 2dB.
The results of RMSE versus snapshots are shown in Fig. 7(b), and a similar conclusion can be obtained that
the RMSE of LR-SDA is significantly lower than those of other five DCAs.
In the third numerical simulation, the number of sources change from 10 to 20.
The results of RMSE versus the number of sources are shown in Fig. 7(c),
where it can be seen that as the number of sources increases, the RMSEs of GENAMS and TS-ENA increase steeply
than those of NADiS, TNA-I, TNA-II and LR-SDA.
Notably, the RMSE of LR-SDA remains the smallest compared to other five DCAs, indicating its superior performance.
\begin{figure}
  \centering
  \subfigure[ RMSE versus SNR]{
    \label{fig:subfig:onefunction}
    \includegraphics[scale=0.17]{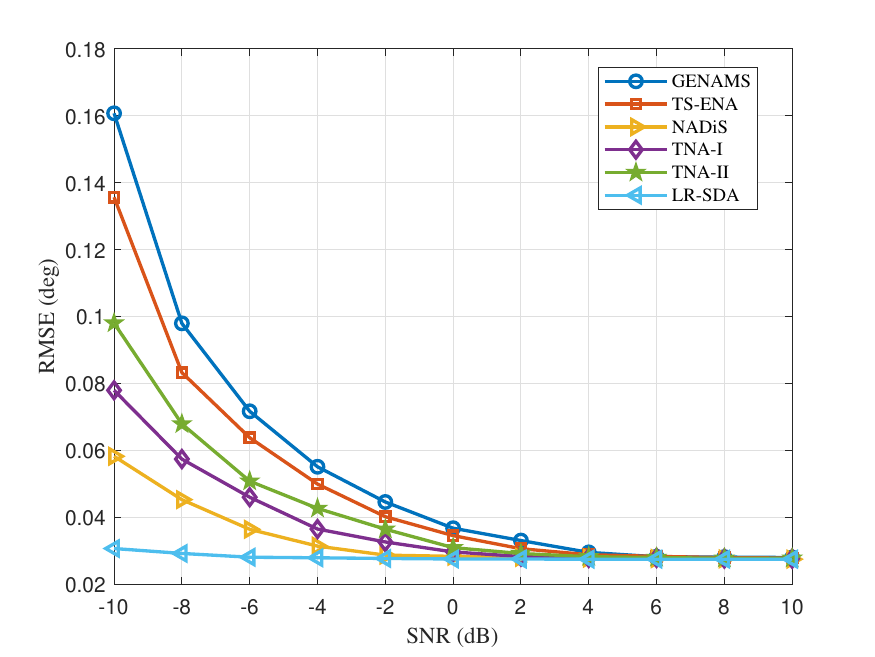}}
   \hspace{0in} 
  \subfigure[RMSE versus Snapshots]{
    \label{fig:subfig:threefunction}
    \includegraphics[scale=0.17]{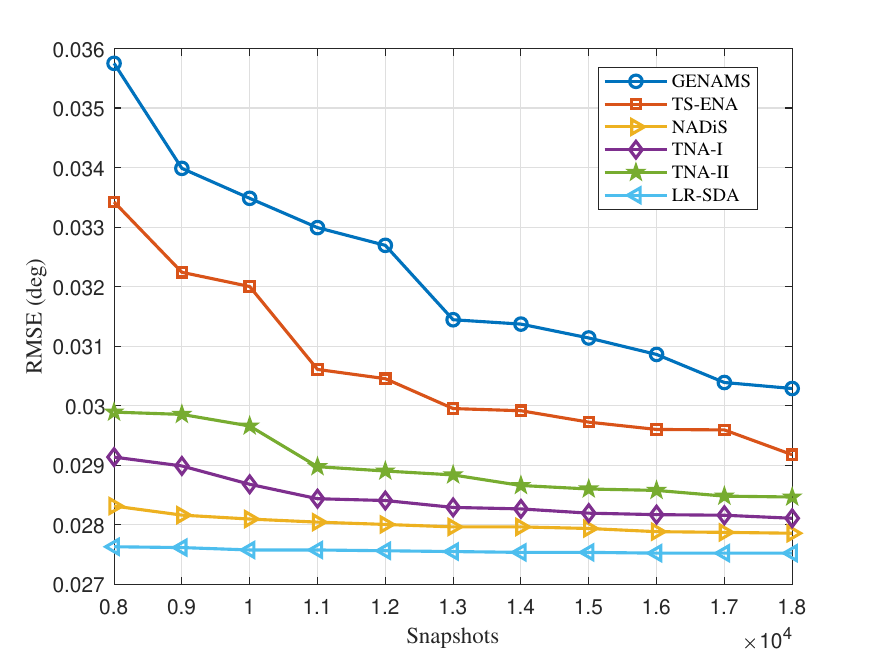}}
    \hspace{0in} 
  \subfigure[RMSE versus Number of Sources]{
    \label{fig:subfig:threefunction}
    \includegraphics[scale=0.17]{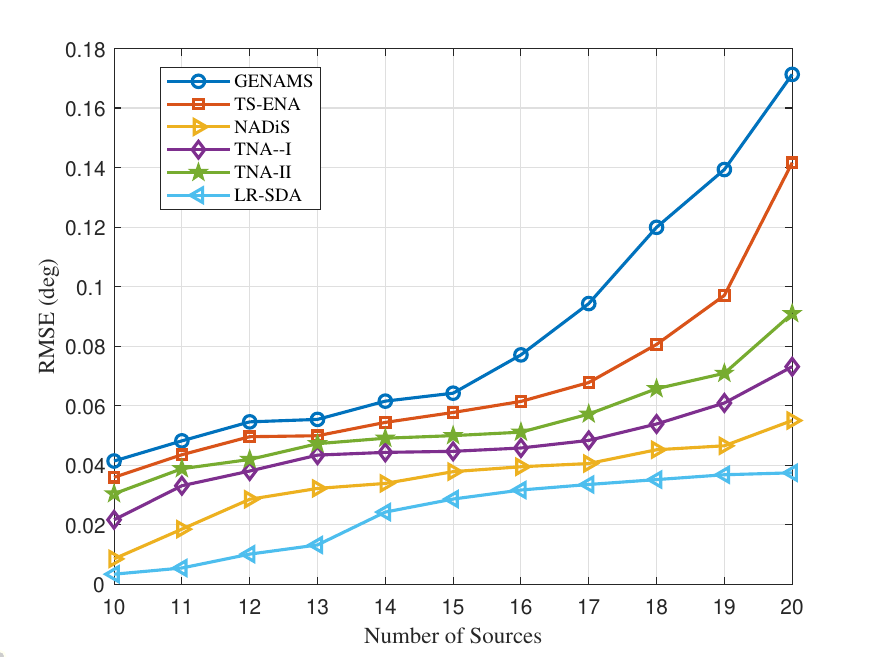}}
  \caption{DOA estimation performance based on second-order cumulants}
\end{figure}

\section{Conclusion}
The paper exploits second-order cumulant to devise the SO-ECA for non-circular signals, whose redundancy is defined in the paper.
Further a novel SLA, namely LR-SDA, is proposed based on SO-ECA,
which can enhance the DOF and improve the DOA estimation performance.
Specifically, LR-SDA consists of three ULAs with a given number of sensors.
Sub-array 1 is a ULA with big inter-spacing between sensors,
and sub-array 2 and 3 are ULA with unit inter-spacing between sensors.
This arrangement can yield closed-form expressions for the sensor positions of the proposed LR-SDA.
Therefore, based on this design, the proposed LR-SDA offers larger DOF than those of other existing DCAs.
The enhanced DOF increase the number of resolvable sources with the given number of physical sensors.
In addition, the necessary and sufficient conditions of signal reconstruction for LR-SDA are derived in this paper.
Further the weight function and redundancy of LR-SDA are defined,
whereas the redundancy of LR-SDA is the lowest compared to other three DCAs designed based on sum-difference co-arrays.

\section*{Acknowledgment}
This work was supported by the National Natural Science Foundation of China (Grant No. 62371400).\\

\appendix

\section{Derivation process of $N_1$ in (\ref{wh7})}

Firstly, according to the problem (P2), $N_2$ can be rewritten as follows when $\eta=1$
\begin{equation}
\label{st7}
N_2=\frac{E-2N_1-2\delta}{2N_1}.
\end{equation}

Substituting (\ref{st7}) into $N$ in (P2), the following equation can be obtained
\begin{equation}
f_1(N_1)\triangleq N_1+N_2=\frac{2N_1^2-2N_1+E-2\delta}{2N_1}.
\end{equation}

In order to know the monotonicity of function $f_1(N_1)$, we solve the first derivation of $f_1(N_1)$ as follows
\begin{equation}
\label{st8}
\frac{\partial f_1(N_1)}{\partial N_1}=\frac{4N_1^2-2E+4\delta}{4N_1^2},
\end{equation}
where the denominator $4N_1^2\geq0$, therefore the sign of the first derivative for $f_1(N_1)$ is determined by the numerator in (\ref{st8}).
Further, the numerator in (\ref{st8}) is a quadratic function, thus the sign is determined by $E$.
To further discuss the positive and  negative polarity of the first derivative of $f_1(N_1)$,
setting $g_1(N_1)\triangleq 4N_1^2-2E+4\delta$.
Since $N_1\geq1$, solving the critical point of $g_1(N_1)$ when $g_1(N_1)=0$ as follows
\begin{equation}
N_1^*=\frac{\sqrt{2}\sqrt{E-2\delta}}{2}.
\end{equation}

When $N_1> N_1^*$, $\frac{\partial f_1(N_1)}{\partial N_1}<0$. At this time, $f_1(N_1)$ is monotonically decreasing.
When $N_1< N_1^*$, $\frac{\partial f_1(N_1)}{\partial N_1}>0$. At this time, $f_1(N_1)$ is monotonically increasing.
Therefore, we can get the minimize value of $f(N_1)$ at the critical point $N_1^*$.

Secondly, when $\eta\geq 2$, according to the problem (P2), $N_2$ can be rewritten as follows
\begin{equation}
\label{st9}
N_2=\frac{E-2N_1-2\delta-2\eta}{2N_1}.
\end{equation}

Substituting (\ref{st9}) into $N$ in (P2), the following equation can be obtained
\begin{equation}
f_2(N_1)\triangleq N_1+N_2=\frac{2N_1^2+E+2-N_1-2\delta}{2N_1+1}.
\end{equation}

In order to know the monotonicity of function $f_2(N_1)$, we solve the first derivation of $f_2(N_1)$ as follows
\begin{equation}
\label{st10}
\frac{\partial f_2(N_1)}{\partial N_1}=\frac{4N_1^2+4N_1-2E+4\delta-5}{(2N_1+1)^2},
\end{equation}
where the denominator $(2N_1+1)^2\geq0$, therefore the sign of the first derivative for $f_2(N_1)$ is determined by the numerator in (\ref{st10}).
Further, the numerator in (\ref{st10}) is a quadratic function, thus the sign is determined by $E$.
To further discuss the positive and  negative polarity of the first derivative of $f_2(N_1)$,
setting $g_2(N_1)=4N_1^2+4N_1-2E+4\delta-5$.
Since $N_1\geq1$, solving the critical point of $g_2(N_1)$ when $g_2(N_1)=0$ as follows
\begin{equation}
N_1^*=\frac{-1+\sqrt{2}\sqrt{E+3-2\delta}}{2}.
\end{equation}

When $N_1> N_1^*$, $\frac{\partial f_2(N_1)}{\partial N_1}<0$. At this time, $f_2(N_1)$ is monotonically decreasing.
When $N_1< N_1^*$, $\frac{\partial f_2(N_1)}{\partial N_1}>0$. At this time, $f_2(N_1)$ is monotonically increasing.
Therefore, we can get the minimize value of $f_2(N_1)$ at the critical point $N_1^*$.

\section{The Proof of Corollary 1}

In order to further investigate the redundancy of LR-SDA,
the relaxed problem of the redundancy $R_w^1(N)=\frac{N^2}{N_1(N_2+1)+(\lceil \frac{N_2}{2} \rceil-1)+\delta}$ is studied.
And the upper and lower bounds of redundancy of LR-SDA when the number of physical sensors
varies from 2 to infinity is analyzed in detail as follows.

Firstly, when there are two or three physical sensors for LR-SDA, the redundancies of LR-SDA are
\begin{equation}\nonumber
R_w^1(2)=1,\ \ \ \  R_w^1(3)=1.125.
\end{equation}

When the number of physical sensors $N\rightarrow \infty$, the redundancy of LR-SDA is
\begin{equation}\nonumber
\begin{aligned}
\lim_{N\rightarrow \infty} R_w^1(N)&=\lim_{N\rightarrow \infty}\frac{N^2}{N_1(N_2+1)+(\lceil \frac{N_2}{2} \rceil-1)+\delta}=2.\\
\end{aligned}
\end{equation}

And the monotonicity of $R_w^1(N)$ is analyzed as follows
\begin{equation}\nonumber
\begin{aligned}
\frac{\partial R_w^1(N) }{\partial N}=\frac{N(4N^2+4N-3)}{2(N_1(N_2+1)+(\lceil \frac{N_2}{2} \rceil-1)+\delta)^2}.
\end{aligned}
\end{equation}

According to the form of $\frac{\partial R_w^1(N) }{\partial N}$, it is can be known that the positive and negative polarities
are determined by the numerator of $\frac{\partial R_w^1(N) }{\partial N}$. Further, for $4N^2+4N-3=0$,
it more than zero when $N\geq 2$, therefore $\frac{\partial R_w^1(N) }{\partial N}>0$,
which means $R_w^1(N)$ strictly increasing.

Therefore, the upper and lower bounds of the redundancy $R_w^1$ with the number of physical sensors varying from 2 to infinity are
\begin{equation}\nonumber
\begin{aligned}
1\leq R_w^1\leq 2.
\end{aligned}
\end{equation}

\end{document}